%% file: main.tex
\documentclass[journal,12pt,onecolumn]{IEEEtran}
\ifCLASSINFOpdf
  \usepackage[pdftex]{graphicx}
\else
\fi
\usepackage{amsmath}
\ifCLASSOPTIONcompsoc
 \usepackage[caption=false,font=normalsize,labelfont=sf,textfont=sf]{subfig}
\else
 \usepackage[caption=false,font=footnotesize]{subfig}
\fi
\input{symbols.tex}
\usepackage{todonotes} 

\usepackage[inline]{enumitem}

\begin{document}
%
\title{\LARGE How biased is  your model? \\Concentration Inequalities,
    Information
     and Model Bias}
%
%
%
\author{Konstantinos~Gourgoulias,
        Markos A.~
    Katsoulakis,
        Luc~Rey-Bellet
        and~Jie~Wang
\thanks{K. Gourgoulias is with Babylon Health, London, UK. He contributed to this work while he was at the Department of Mathematics and Statistics, University of Massachusetts Amherst, Amherst, MA, 01003,  USA (e-mail: kostis.gourgoulias@babylonhealth.com) }
\thanks{M. Katsoulakis is with the Department
of Mathematics and Statistics, University of Massachusetts Amherst, Amherst,
MA, 01003, USA (e-mail:markos@math.umass.edu) }
\thanks{ L. Rey-Bellet is with the Department
of Mathematics and Statistics, University of Massachusetts Amherst, Amherst,
MA, 01003, USA (e-mail:luc@math.umass.edu) }
\thanks{ J. Wang is with the Department
of Mathematics and Statistics, University of Massachusetts Amherst, Amherst,
MA, 01003, USA (e-mail:wang@math.umass.edu) }
}

\maketitle

\begin{abstract}
We derive tight and computable bounds on the bias of statistical estimators, or more generally of quantities of interest,  when evaluated on a baseline model $P$ rather than on the typically unknown  true model $Q$. Our proposed method  combines  the  scalable information inequality derived by P. Dupuis, K.Chowdhary, the authors and their collaborators together with classical concentration inequalities (such as Bennett's and Hoeffding-Azuma inequalities). Our bounds are expressed in terms of the Kullback-Leibler divergence $\R QP$ of model $Q$ with respect to $P$ and the moment generating function for the statistical estimator under $P$. Furthermore, concentration inequalities, i.e. bounds on moment generating functions,  provide tight and computationally inexpensive    model bias bounds for quantities of interest. Finally, they allow us to derive rigorous  confidence bands for statistical estimators that account for  model bias and are valid for an arbitrary amount of data.
\end{abstract}

\begin{IEEEkeywords}
Uncertainty quantification, information theory, information bounds, model bias, model uncertainty, goal-oriented divergence, concentration inequalities, Kullback-Leibler divergence, statistical estimators
\end{IEEEkeywords}

%
\IEEEpeerreviewmaketitle

\section{Introduction}
%
%
%
%

 

\label{sec:intro}
An essential ingredient of predictive modeling is the reliable calculation  of specific
statistics/quantities of interest  of the predictive distribution. Such
statistics are typically tied to the application domain, for instance moments,
covariance, failure probabilities, extreme events, arrival times, average velocity, energy
and so on. Predictive models can involve (a) statistical aspects or data
collection, and (b) physical/mathematical mechanisms with choices in
complexity/resolution, some of them potentially computationally intractable.
Therefore, to improve the predictive capabilities of models we face  fundamental trade-offs
between model complexity,  amount of available data,  computational efficiency, and  model bias.  

The main focus of the paper is the understanding and control of model bias which often inevitably occurs in model building and which is itself a measure  of reliable predictions. 
Our primary tool are information-theoretic Uncertainty Quantification  methods.
Uncertainty quantification (UQ) methods address questions related to model selection, model sensitivity, model reduction and misspecification, \cite{Dakota,Smith,Saltelli:global}. 
Sources of uncertainty are broadly classified in two categories: aleatoric, due to the inherent stochasticity of probabilistic models and the limited availability of data, and epistemic, stemming from the inability to accurately model all aspects of a complex system, \cite{dupuis2011uq, Smith, Sullivan2015}.
%
Model bias is closely related to epistemic uncertainty, and probability metrics (Wasserstein, total variation) and divergences (Kullback-Leibler, Renyi, $\chi^2$) \cite{tsybakov} are important tools to quantify uncertainty by comparing  models.
%
Among the divergences, the Kullback-Leibler (KL) divergence (also known as  relative entropy) is widely used 
because of its computational tractability.
Specifically, KL-based methods have been used successfully in variational inference and expectation propagation~\cite{bishop}, model selection~\cite{burnham}, model reduction (coarse-graining)~\cite{Shell, KP2013, Noid,  majdabook}, optimal experiment design, \cite{Atkinson2007OptimumExperimentalDesignsWithSAS.}, and UQ~\cite{majdauncertainty,majdafidelity,Komorowski,PK2013}.

Information-theoretic methods
for model building
will typically induce bias for  the various statistics and the QoIs of the predictive distribution compared to the ``true'' model--if known--or the available data.
Managing the corresponding trade-offs between a range of less biased but more computationally expensive models naturally leads to the following main question 
for the paper :
\begin{center}
  \textit{Can we provide performance guarantees for model bias in models built
    via KL-based approximate inference, model misspecification, or model
    selection methods?}
\end{center}

In this paper we ultimately seek to understand how a decrease in KL-divergence--associated with an increase in modeling and/or computational effort--can
guarantee a model bias tolerance; and in addition, we seek the tightest possible control of model bias.
Note that bounds on the model bias of a QoI between two distributions $P$ and $Q$
can be obtained, for example, in terms of their KL or $\chi^2$ divergences using the classical Pinsker or Chapman-Robbins inequalities respectively, \cite{Cover2006ElementsofInformationTheory,tsybakov}. Clearly a decrease in divergence will improve bounds on the model bias.
However, these classical inequalities are 
typically \textit{non-tight} and
\textit{non-discriminating}, in the sense that they scale poorly with the size of data sets,  with the number of variables in high-dimensional models (e.g. molecular systems), or with time in the context of stochastic  processes; we refer to  Sections~2.2--2.3 in \cite{katsoulakisjcp2017} for a complete discussion, see also the  example in Remark~\ref{rem:poor:scalab}.


To tackle these challenges a class of new information  inequalities have been introduced by Paul Dupuis in \cite{dupuis2011uq} and  further developed in \cite{katsoulakisuq16,katsoulakisjcp2017}
by the authors and their collaborators. 
The resulting bounds on model bias bounds involve (a) the KL divergence $\R QP$
between a baseline model $P$ and  an alternative  models $Q$, and (b) the moment generating function (MGF) for the QoI under  the baseline model  $P$. This inequality inherits the asymmetry of $\R QP$, which in turn allows us to exchange the roles of $P$ and $Q$,  depending on the context and/or availability of data from either $P$ or $Q$. 
Considering a neighborhood of models around the baseline $P$, defined by the KL divergence $\R QP$,  can be associated with a specified error tolerance and is non-parametric in nature.
The crucial mathematical ingredient behind the inequality is the Donsker-Varadhan variational principle~\cite[Appendix C.]{dupuisellis} for the KL divergence, also known as the  Gibbs variational
formula~\cite{simon2014statistical}.
This variational representation actually implies that 
the new inequalities are \textit{tight}, i.e. they become an equality for a suitable model $Q$ within a  KL divergence neighborhood of the baseline model $P$.
Furthermore, the dependence on   the MGF renders the bounds scalable and \textit{discriminating} for high-dimensional data sets and models, e.g. Markov Random Fields, long-time dynamics of stochastic processes and molecular models, \cite{katsoulakisjcp2017} as demonstrated recently in  \cite{katsoulakisjcp2017}.
Finally, broadly related methods  in model misspecification and sensitivity analysis in  financial risk measurement and queuing theory, using a robust optimization perspective, were proposed recently  in    \cite{glasserman2014} and \cite{lam2016}, we also refer to references therein for other  related work  in operations research, finance  and macroeconomics.

The primary goal of this paper is to use these new theoretical advances to develop practical tools to estimate and control model bias, and this raises new theoretical questions and implementation challenges. In particular evaluating or estimating MGFs can be very costly due to high variance of the estimators thus requiring either a large amount of data, see also Table 1, or multi-level/sequential  Monte Carlo methods~\cite{junliu,bishop,lelievre2010free,delmoral2012}. In this paper we rather pursue the use of a variety of QoI-dependent \textit{concentration  inequalities}~\cite{concentration,raginsky2013concentration,ledoux2005concentration}
to bypass the evaluation or estimation of the MGF and this leads to computable, tight bounds for model bias. 
%
%
Concentration inequalities are a fundamental mathematical tool in the study of rare events~\cite{dembo2010large}, model selection
methods~\cite{massart2007concentration}, statistical mechanics~\cite[Section
8.4]{ledoux2005concentration}, random matrix
theory~\cite{tropp2015introduction}. Usually concentration inequalities are used to bound tail events, i.e. to provide bounds on the probability that a random variable deviates from typical behavior.
In this paper we use concentration inequalities for
the purpose of uncertainty quantification,  specifically to control model bias,  by implementing efficiently the new information inequalities developed in \cite{dupuis2011uq,katsoulakisuq16,katsoulakisjcp2017}, while at the same time maintaining and expanding their theoretical advantages. 



The new inequalities proved in this paper--- which we call  \textit{concentration/information} inequalities---combine  concentration inequalities with the variational principles  underlying the bounds and  lead to model bias bounds with the following key features:  
\begin{enumerate}
\item[(a)] Easily  computable bounds  in terms of simple properties of the QoIs such
    as their mean,  
    upper and lower bounds,  suitable bounds on their variance, and so on; that is, without requiring the costly computation of MGFs. 
    \item[(b)] Scalability for  QoIs that depend on large numbers of data such as statistical estimators,  or for high dimensional probabilistic models. 
    \item[(c)]   Derivation of  rigorous  confidence bands for statistical estimators that account for  model bias and are valid for an arbitrary amount of data.   
    \item[(d)] Applicability  to families of QoIs satisfying  a concentration inequality, 
    and not to just a single QoI.
    \item[(e)] Tightness of  the model bias bounds in the sense that the bounds are always attained within a prescribed KL-divergence \textit{and} the  class of QoIs in (d). 
\end{enumerate}

The structure of the paper is as follows.  In Section~\ref{sec:outline} we set-up the  mathematical  framework for the paper and discuss 
the information inequalities for QoIs of \cite{dupuis2011uq,katsoulakisuq16,katsoulakisjcp2017}.
In Section~\ref{sec:conc-ineq} we use concentration inequalities to derive new concentration/information inequalities on model bias that are typically  straightforward to implement. In Section~\ref{sec:sharpness-and-robustness}, we discuss the  tightness properties of the new concentration/information bounds. 
Finally in Section~\ref{sec:scalability}  we study the bias of statistical estimators, noting that such QoIs will require results  that scale properly with the amount of available data.
We also illustrate the bounds in a variety of examples. In  Section~\ref{ex:simple-examples} 
we  consider two elementary examples  with bounded or unbounded QoIs. Two examples of systems with epistemic uncertainty are discussed in Section~\ref{sec:epistemic-UQ-conc}; the first one deals with failure probabilities for batteries 
and the second with  Markov Random Fields such as Ising systems.

\section{Tight Model Bias Bounds using KL Divergence}
\label{sec:outline}

In coarse-graining, model reduction, 
model selection, or variational inference,
as well as in other uncertainty quantification and approximate inference problems, a baseline model $P$ is compared  to a "true" or simply a different model $Q$. In this case the notion of {\em risk or mean
  square error} plays a key role in assessing the quality of the corresponding
estimators. Namely, if $\hat{f}$ is an unbiased estimator of the quantity of
interest $f$ for the baseline  model $P$ (but not of the "true" model $Q$)
then the risk of the estimator is the mean squared error 
\begin{align}
  \label{eq:risk}
  \mbox{RISK}:=\mathbb{E}_P[ (\hat{f} - \mathbb{E}_Q[f])^2]= \underbrace{{\rm
  Var}_{P}[\hat{f}] }_{{\rm Variance}}+ \underbrace{ | \mathbb{E}_P[f] -
  \mathbb{E}_Q[f]|^2}_{\rm Model Bias} \,.
\end{align}
If available computational resources can be used to control the variance of
the baseline model P, then the model bias becomes the dominant source of risk thus must be carefully controlled. The main  goal of this work is to understand how to transfer quantitative results on information metrics, specifically the KL divergence  $R(Q\|P)$ (also known as relative entropy), to {\em bounds on the bias for
  quantities of interest} $f$. We formulate the corresponding mathematical
problem next.

\medskip
\noindent{\bf Mathematical Formulation.} 
Let us consider a baseline model given 
by the probability measure $P$ on the 
state space $\mathcal{X}$ which we assume to be a Polish (i.e. complete separable metric) space and  we consider 
a QoI $f$, that is a measurable function 
$f: \mathcal{X} \to\mathbb{R}$. 
We specify next a family of alternative 
probability distributions in terms of the 
Kullback-Leibler (KL) divergence (or relative entropy) $\R Q P$, which is defined as 
\begin{align}
  \label{eq:KL}
  \R QP=\int \log \frac{dQ}{dP} dQ. 
\end{align}
if $Q$ is absolutely continuous with respect to $P$ 
(and $+\infty$ otherwise).  Note that $\R QP$ the 
properties of a divergence that is  $\R QP\geq 0$ for 
all $Q$ and  $\R QP=0$ if and only if  $Q=P$, see e.g. 
\cite{Cover2006ElementsofInformationTheory}.

We fix a positive number $\eta$ which we interpret 
as a level of \textit{model misspecification}, quantified in terms KL divergence or, alternatively, as an error tolerance level between the baseline  model $P$ and alternative models described $Q$.
We then define the set of alternative probability as 
\begin{equation}
  \label{Qfamily}
  \mathcal{Q}_\eta=\{Q: R(Q\|P) \le \eta^2\}\, .
\end{equation}
and any $Q\in\efam$ is referred to as an \textit{$\eta$-admissible model}.
We remark that our approach is 
\textit{non-parametric}, i.e. it does not
rely on any parametric form of the 
probability distributions considered. 
The relative entropy $\R QP$ is convex and 
lower-semicontinuous in $(Q,P)$. 
In general the set  $\efam$ is infinitely  dimensional, although it is compact with respect to the weak topology, \cite{dupuisellis}. 
The fact that the KL divergence is not symmetric in its arguments can be advantageous in some situations. For 
example, in variational inference, it naturally imposes 
a constraint on the support of the possible 
approximations $Q$ of a target model $P$~\cite{bishop}.

Our primary mathematical challenge in 
this work lies in quantifying the model
bias in~(\ref{eq:risk}) if we use an $\eta$-admissible model in $Q_\eta$ rather than the baseline model $P$. That is we need to 
\begin{align*}
\textrm{ Compute (or estimate) }\sup_{Q \in \efam} \{\mathbb{E}_Q[f] - \mathbb{E}_P[f] \} \quad \textrm{ and } \quad \inf_{Q \in \efam} \{\mathbb{E}_Q[f] - \mathbb{E}_P[f]\}\,.
\end{align*}
Note that this approach is intrinsically 
goal-oriented since it includes not only  a 
family of alternative models $Q$ but also  a specific choice of QoI $f$.

\medskip
\noindent{\bf Goal-oriented divergence.}
We now define a divergence which incorporates the QoI $f$ and hence is called goal-oriented; it was first introduced in the current form in~\cite{katsoulakisuq16} based on earlier work in~\cite{dupuis2011uq}.
Consider a QoI $f$ and the moment-generating function (MGF)
\begin{align}
  \label{eq:mgf}
  \Mp{c;\ti f}:= E_{P} [e^{c\ti f}]
\end{align}
of the centered QoI $\ti f$,
\begin{align}
  \label{eq:centralized}
  \ti f(x) := f(x)-\Ep[f].
\end{align}
In general (see \cite{dembo2010large} for details) 
the MGF $\Mp{c;\ti f}$ is finite for $c$ in some 
interval $I$ and equal to $+\infty$ otherwise.  
Throughout this  paper we will make the standing 
assumption that $\Mp{c;\ti f}$ is finite in the interval  $I=(d_-, d_+)$ with $d_- < 0 < d_+$, then under this assumption, $\Mp{c;\ti f}$ is $C^\infty$ in $I$ and 
$f$ has finite moments of any order.  We next define  the goal-oriented (GO) divergence as 
\begin{equation}
  \go QPf = \inf _{c>0}\left\{ \frac{1}{c}  \Cp{c;\ti f} +  \frac{1}{c} \R QP \,  \right\}.\label{eq:intro:godiv:2} 
\end{equation}
for $P,Q$ with $\R Q p < \infty$. 
Note that if $d_+$ is finite then  
the infimum can be taken on $(0,d_+)$
and note also that if  $\R Q P = \infty$ 
then the goal  oriented divergence can 
naturally be then set equal to $+\infty$.

In \cite{dupuis2011uq,katsoulakisuq16} 
the following bound on the model bias was proved, along with certain mathematical properties:
\begin{theorem}
\label{thm:GO-bounds} Let $P$ be a probability measure and let $f$ be such that its MGF $\Mp{c;\ti  f}$ is finite in a neighborhood of the origin. Then for any $Q$ with $\R Q P < \infty$ we have 
\begin{equation}
  -\go QP{-f} \leq \bias \leq \go QP{f}.
  \label{eq:intro:bias-bound-go:sec2}
\end{equation}
The GO divergence has the following properties
\begin{enumerate}
\item \label{it:go:non-negative} \textit{Divergence:}
$\Xi(Q \mid \mid P ;f ) \ge 0$ 
and $\Xi(Q \mid \mid P ;f) = 0$ if and only if either $Q=P$
  or $f$ is constant $P$-a.s.
\item \label{it:go:linearization} \textit{Linearization}: 
$$ \Xi(Q \mid \mid P;\pm f)=
\sqrt{\Vp[f]}\sqrt{2R(Q\mid \mid P)}+O(\R QP)$$ and thus
$$
\mid E_Q(f)- E_P(f)\mid \, \le \, \sqrt{\Vp[f]}\sqrt{2R(Q\mid \mid P)}+O(R(Q\mid
\mid P)).
  $$
\end{enumerate}
\end{theorem}

\medskip
\noindent{\bf Tightness of goal-oriented divergence.}
Our next result complements Theorem~\ref{thm:GO-bounds}
and demonstrates the tightness of the GO divergence
bounds \eqref{eq:intro:bias-bound-go:sec2} for the bias of a  QoI $f$; for the complete proof, we refer to Appendix~\eqref{sec:Appendix:tightness:GO}. An     equivalent tightness result for the upper bound in   \eqref{it:go:tightness:thm} was first shown in \cite{dupuis2011uq}, while   here we present a new formulation  for a  complete  tightness result  in Theorem~\ref{thm:tightness},    based on the goal-oriented divergence formulation in \cite{katsoulakisuq16}.

To state our result we introduce 
the exponential family $P^{c}$ given by 
\begin{equation}
  \label{eq:parametric}
  \frac{dP^{c}}{dP} \,=\, e^{ c f - \log M_P(c; f)}\,=\,\frac{e^{c f}}{\int e^{c f}dP}
  \,,
\end{equation}
which is well-defined for $c$ in the 
interval $I=(d_-,d_+)$ where $\Mp{c;f}$ is finite.

\begin{theorem}
  \label{thm:tightness} 
Let $P$ be a probability measure and let 
$f$ be such the MGF $\Mp{c;\ti  f}$  is finite in a neighborhood of the origin.
Let $\mathcal{Q}_\eta=\{Q: R(Q\|P) \le
\eta^2\}$ be the set all  approximate 
probability $Q$ within a KL tolerance $\eta^2$. 
\begin{enumerate} 
\item There exists $0< \eta_{\pm} \le \infty$ such that for any $\eta \le \eta_{\pm}$ there are probability measures $Q^{\pm}$ such that  
\begin{equation}
\label{eq:max}
\begin{aligned}
\go {Q^+}Pf &= \E_{Q^+}[f] -\E_{P}[f] &=
                \max_{Q\in \mathcal{Q}_\eta}\E_{Q}[f]-\E_{P}[f],
\end{aligned}
\end{equation}
\begin{equation}
\label{eq:min}
\begin{aligned}
-\go {Q^{-}}P{-f}& = \E_{Q^-}[f] - \E_{P}[f] &=
                   \min_{Q\in \mathcal{Q}_\eta}\E_{Q}[f]-\E_{P}[f]
\end{aligned}
\end{equation}
The measures $Q_\pm$ are given by the elements $P^{c_\pm}$ of the exponential family \eqref{eq:parametric} where $c_\pm$  are the unique solution of $\R {P^{c_\pm}}{P} = \eta^2$. 
\item If $\eta_\pm$ is finite then $f$ is necessarily bounded above/bounded below 
$P$-almost surely with upper/lower bound $f_{\pm}$. In that case, for $\eta > \eta_{\pm}$ and any $Q$ with $\R QP = \eta^2$
we have  
\begin{equation}
 \label{eq:max2}
\begin{aligned}
\go QPf &= f_+ -\E_{P}[f]&=\sup_{Q\in \mathcal{Q}_\eta}\E_{Q}[f]-\E_{P}[f],
\end{aligned}
\end{equation}
\begin{equation}
\label{eq:min2}
\begin{aligned}
-\go QP{-f}& = f_- -\E_{P}[f]&=
                    \inf_{Q\in \mathcal{Q}_\eta}\E_{Q}[f]-\E_{P}[f]
\end{aligned}
\end{equation}
\end{enumerate}
\end{theorem}


The main result of the theorem provides 
performance guarantees in the sense that $\E_Q[f]$ belongs to the interval 
\begin{equation}
\label{it:go:tightness:thm}
\begin{aligned}
  -\go {Q^{-}}P{-f}+&\E_{P}[f]
  \le
  \E_{Q}[f]  
  \le & 
  \E_{P}[f]+\go{Q^+}Pf\,  \quad \mbox{for all  $Q\in \mathcal{Q}_\eta$ } 
\end{aligned}
\end{equation}
and the bounds are tight in $\mathcal{Q}_\eta$, in the sense 
that inequalities become equalities  for $Q=Q^\mp$ respectively. 
%
%
%
%
This tightness property is crucial for our discussion because it implies that the GO divergence bounds in~(\ref{eq:intro:bias-bound-go:sec2}) are the best possible in the sense that they have attainable worst-case model scenarios $Q^\pm$ among all probability distributions $Q$ within a KL tolerance
$\eta^2>0$, see the schematic in Figure~\ref{fig:tightness_GO}. 
\begin{figure}[!ht]
  \centering \includegraphics[width=0.5\textwidth]{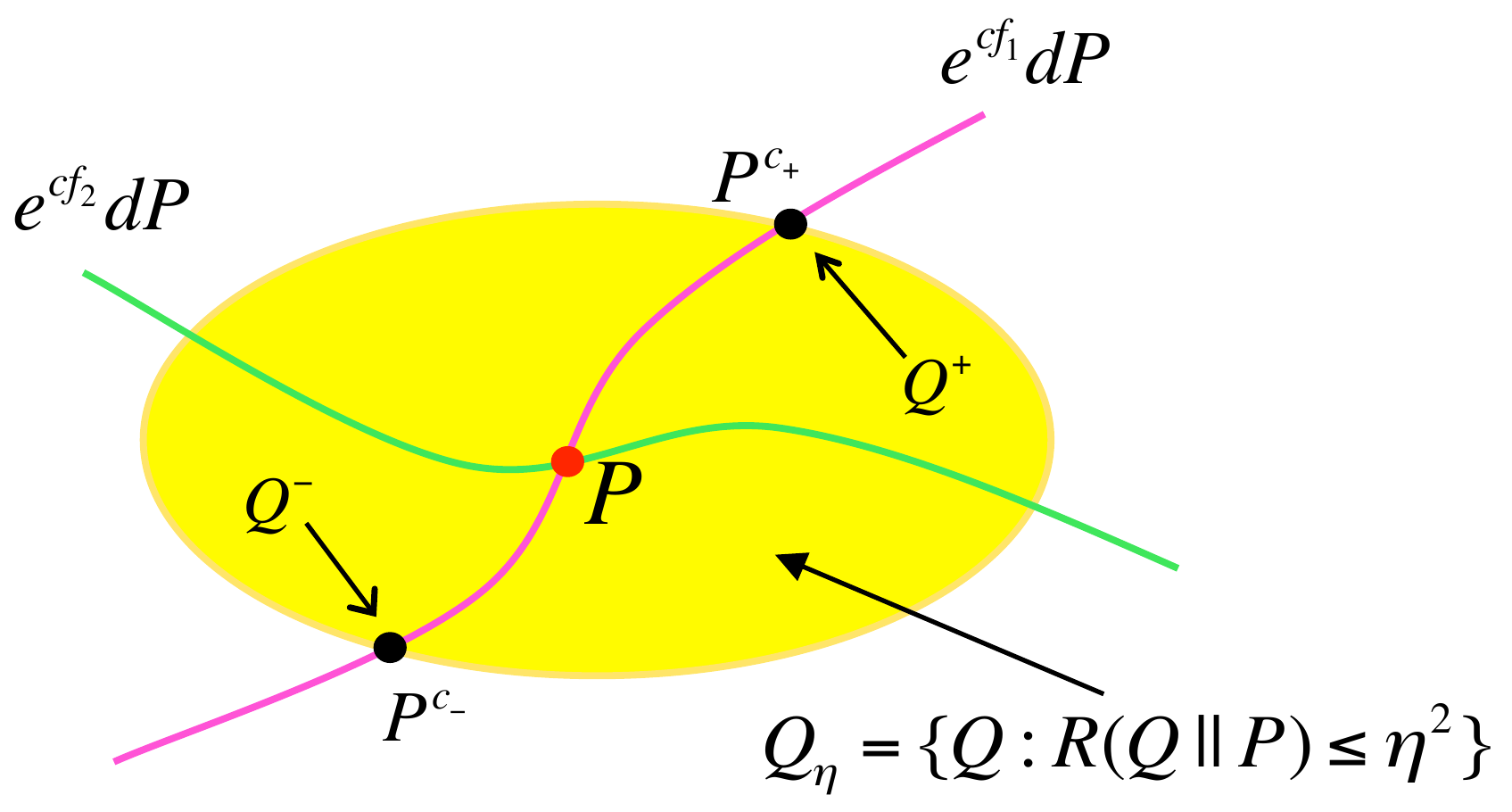}
  \caption{ The schematic depiction of Theorem~\ref{thm:tightness} for the QoIs $f_1$, $f_2$ with tolerance $\eta^2$. The solid lines depict the one-parameter tilted
    probability distributions $P^c$ corresponding to the QoIs. The theorem
    implies that the upper and lower bounds in the family
    $\mathcal{Q}_\eta=\{Q: R(Q\|P) \le \eta^2\}$ are attained at the
    probability measures $Q^{\pm}=P^{c_{\pm}}$ for the QoI $f_1$. }
  \label{fig:tightness_GO}
\end{figure}

\begin{remark}{\rm 
The tightness
property~(\ref{it:go:tightness:thm}) is a \textit{non-parametric} result: the family $\mathcal{Q}_\eta$
of all alternative models Q  cannot be parametrized in general and is only
characterized by the property $\R QP\le \eta^2$. In spite of this non-parametric
framework, we showed in Theorem~\ref{thm:tightness}
that the extremal models $Q^\pm$ that yield the tight bounds \eqref{it:go:tightness:thm}
belong to the parametrized family \eqref{eq:parametric},  see also Figure~\ref{fig:tightness_GO}. 
}
\end{remark}
 \begin{table}
  \centering
    \caption{
    \label{tab:estimator-cost}
    For the estimation of $\Vp[Y]$, we assume that $\Ep[Y]$ is unknown and that
    the bias-adjusted estimator is used. For the variance of $\Ep[e^{cY}]$, a
    first-order approximation is used (see \cite{wasserman2004}),
    assuming that $\Ep[Y]$ is small.}
    
    \begin{tabular}{|c||c|}
    \hline
      Quantity& Variance of estimator\\
      \hline 
      $E_P[Y]$& $\Vp[Y]/n$\\
      \hline
      $\Vp[Y]$& $2(\Vp[Y])^2/(n-1)$\\
      \hline
      $\Mp{c;Y}$& $c^2e^{2c\Ep[Y]}\var_P[Y]/n$\\
      \hline
    \end{tabular}
  \end{table}
The  attractive properties of the GO bounds demonstrated in Theorem~\ref{thm:GO-bounds} and 
Theorem~\ref{thm:tightness}, come at a potentially significant cost since they require the knowledge or calculation of the MGF $\Mp{c;\ti f}$ with respect to model $P$. 
If no simple formula for $\Mp{c;\ti f}$ is known,   this can be  a data-intensive operation---compare the estimator variance of the MGF with that of other QoIs in Table~\ref{tab:estimator-cost}. 
Controlling the variance of  an MGF estimator will require a large amount of data and/or the use of a
  multi-level Monte Carlo method, see also the discussion in
  Section~\ref{sec:intro} and 
  Section~\ref{sec:epistemic-UQ-conc}. 

In the next section we introduce a new class of inequalities that share the aforementioned  features of the GO divergence and \eqref{eq:intro:bias-bound-go:sec2}, 
but they can  bypass the estimation of an MGF by using the concept of concentration inequalities.

  \section{Concentration/Information Inequalities  for Model Bias}
  \label{sec:conc-ineq}
  To bypass the estimation or 
  computation of the MGF in \eqref{eq:intro:godiv:2} we will  use a QoI-dependent concentration bound for the MGF, \ie a function $\Phip{c}$
taking values in $(0, \infty]$ 
such that 
  \begin{align}
    \label{eq:conc:mgf-bound}
    \Mp{c;\ti f}\leq \Phip{c} 
  \end{align}
for all $c \in \mathbb{R}$. 
Since the moment generating function 
$\Mp{c;\ti f}$ can take the 
value $+\infty$ it is natural to 
allow the same for $\Phip{c}$. 

Bounds of the form 
\eqref{eq:conc:mgf-bound},
 for explicitly computable 
 functions $\Phip{c}$,
are called concentration inequalities  and we discuss  several such examples in Section~\ref{sec:unbounded-obs} and Section~\ref{sec:bounded-obs}, as well as in Section \ref{sec:scalability}. 
Although we use only the simplest concentration inequalities here, the results are indicative  to what can be accomplished using such 
  information on $f$ and $P$. In upcoming work, we will consider further applications for 
  stochastic  processes and interacting particle systems arising in Kinetic Monte Carlo and molecular dynamics models.
Concentration inequalities is an important  mathematical tool
since they allow, via a Chernov bound, to control  tail events, i.e. they provide explicit bounds on the probability that a
random variable deviates from typical behavior. More specifically, such methods  can address, among others, questions on rare events~\cite{dembo2010large}, model selection
methods~\cite{massart2007concentration}, statistical mechanics~\cite[Section
8.4]{ledoux2005concentration} and random matrices~\cite{tropp2015introduction}.  
Here we propose the use of concentration inequalities in tandem with the information inequalities \eqref{eq:intro:bias-bound-go:sec2}
for uncertainty quantification and especially for providing model bias guarantees.
In  Theorem~\ref{thm:generalized-go-bounds} we show how to construct 
new bounds for the model bias using a function  $\Phi$ satisfying (\ref{eq:conc:mgf-bound}).

\begin{theorem}
  \label{thm:generalized-go-bounds}
  Let $P$ be a probability measure and let $f$ be a QoI such that its MGF $\Mp{c;\ti  f}$ is finite in a neighborhood of the origin. 
%
 %
Let $\Phi:\mathbb{R}\to  (0,\infty]$  be a  function with $\Phip{0}=1$, $\Phi'(0)=0$ and such that 
  \begin{equation}
  \label{eq:conc:mgf-bound:thm}
  \Mp{c;\ti f}\leq \Phip{c}
  \end{equation}
 for all $c \in \mathbb{R}$.
%
  We define the set of admissible QoIs by
  \begin{align}
    \label{eq:Fp}
    \Fp=\{g:\Mp{c;\tilde g}\leq \Phip{c} 
    \}.
  \end{align}
Then,  $f \in \Fp$,  and for every $Q\in \efam =\{Q: R(Q\|P) \le
\eta^2\}$  we have
\begin{equation}
 \label{eq:U-bounds-bias}
 \begin{aligned}   
    -\Um \eta\Fp &\leq \Eq[g]-\Ep[g]  &\leq \Up \eta\Fp\, \quad  \mbox{for every $g\in\Fp$}\, ,
    \end{aligned}
\end{equation}
    where 
    \begin{equation}
    \Upm \eta\Fp:=\inf_{c>0}\left \{\frac{1}{c}\log\Phi(\pm c)+\frac{1}{c} \eta^2 \right\}. \label{eq:U-divpm}
  \end{equation}
\end{theorem}
\begin{proof}
The proof follows immediately from \eqref{eq:intro:godiv:2} and \eqref{eq:intro:bias-bound-go:sec2}, combined with the concentration inequality \eqref{eq:conc:mgf-bound:thm} and the definition of the admissible QoIs, $\Fp$. We discuss specific examples of inequalities of the type  \eqref{eq:conc:mgf-bound:thm} and their corresponding admissible sets $\Fp$, in Section~\ref{sec:unbounded-obs} and Section~\ref{sec:bounded-obs} below.
\end{proof}



\begin{remark}[Admissible set of QoIs]{\rm
\label{rem:admiss}
We note that the function 
$\Phi$ depends both on the QoI 
$f$ and on $P$ through 
\eqref{eq:conc:mgf-bound:thm}
and therefore the  set of 
admissible functions $\Fp$ also 
depends on the QoI $f$ and on 
$P$. However, to keep notation simple, we suppress this dependence for both $\Phi$ and $\Fp$.
}
\end{remark}

\begin{remark}
[Computing $\Upm \eta\Fp$]{\rm   
Some concentration bounds  \eqref{eq:conc:mgf-bound:thm} such as the sub-Gaussian and  Hoeffding bounds discussed below provide explicit formulas for $\Upm \eta\Fp$, see for instance \eqref{eq:sub-G-Upm} and (\ref{eq:hoef-u}). However, in general---see the sharper Bennett bounds in~(\ref{eq:bennet-mgf-ineq} and
(\ref{eq:bennet-(a,b)})---we have an explicit formula for $\Phi$ but no explicit closed form solution of the optimization over $c$. 
The elementary one-dimensional optimization in \eqref{eq:U-divpm} can be carried out with standard solvers, \eg Newton's method.
}
\end{remark}

\sepeq{Divergence structure of $\Upm \eta\Fp $:}
The following properties of the bounds  $U_{\pm}$ in 
\eqref{eq:U-divpm}
are analogous to the properties of the GO divergence \eqref{eq:intro:godiv:2} outlined in Theorem~\ref{thm:GO-bounds}. One notable difference is that here the divergence structure defined by $\Upm \eta\Fp $ contains information about the entire family $\Fp$ in \eqref{eq:Fp}
and not just a single QoI $f$ as was the case in the GO divergence \eqref{eq:intro:godiv:2}.

\begin{theorem}\label{thm:div-U}
Under the assumptions of Theorem~\ref{thm:generalized-go-bounds} and, in addition, if 
\begin{equation}    \label{eq:tightness:U:condition:thm:div-U}
    \Phi(c)=M_{\bar P}(c; \tilde h)\, ,  
  \end{equation} 
 for  some probability $\bar P$ and QoI $h$ and all $c \in \mathbb{R}$ then $\Upm\eta\Fp$ satisfy: 
\begin{enumerate}
\item {\it Divergence Properties:} \label{it:u:non-neg} 
\begin{itemize}
\item[a.] $\Upm \eta\Fp \geq 0$,  and 
\item[b.] $\Upm \eta\Fp=0$ if and only if $\eta=0$ or $\Fp$ is trivial, i.e. consists only  of functions which are constant $P$-a.s. 
\end{itemize}
\item \label{it:u:linear}{\it Linearization:}  If $\Phi=\Phi(c)$ is twice differentiable in a neighborhood of $c=0$, then we have the asymptotics  
$\Upm \eta\Fp=\sqrt{2\Phi''(0)}\eta +O(\eta^2)$ and thus,
\begin{equation}
\label{eq:linearize-U}
 \begin{aligned}
    |\Eq[g]-&\Ep[g]| \leq \sqrt{2\Phi''(0)}\eta +O(\eta^2)\,   &\mbox{for all $g \in \Fp$} \quad \mbox{and  all $Q\in \efam$}\, .
 \end{aligned}
\end{equation}
\end{enumerate}
\end{theorem}

\begin{proof} The proof follows from Theorem~\ref{thm:GO-bounds}.  Indeed since, by assumption, $\Phi(c)=M_{\bar P}(c; \tilde h)$ we have 
\begin{equation}
\Upm \eta\Fp = \go Q{\bar P}{\pm h} 
\end{equation} 
for any probability $Q$ such that $\R Q {\bar P}=\eta^2$.  Therefore, by Theorem~\ref{thm:GO-bounds}, $\Upm \eta\Fp \ge 0$ and  
$\Upm \eta\Fp = 0$ if and only if  $\eta =0$ or $h$ is constant ${\bar P}$ a.s.  But if  
$h$ is constant ${\bar P}$ a.s then $\Phi(c)=M_{\bar P}(c; \tilde h)=1$ for all $c$ and thus the set of admissible QoIs \eqref{eq:Fp} becomes: 
\begin{align}
    \label{eq:Fp:app}
    \Fp=\{g:\Mp{c;\tilde g}\leq \Phip{c}=1
    \}\, .
  \end{align}
However for any $g \in \Fp$, 
by Jensen's  inequality, $\Mp{c;\tilde g}\ge 1$ since $\mathbb{E}_{P}[\tilde g]=0$.  Therefore the admissible set $\Fp$ consists only  of constant functions thus $g$ is constant $P$-a.s. Finally, the linearization in Theorem~\ref{thm:div-U} is proved similarly to the linearization result 
of the GO divergence in Theorem~\ref{thm:GO-bounds},  (see the proof in Section 3 of \cite{katsoulakisuq16}).
\end{proof}

Theorem~\ref{thm:generalized-go-bounds} 
and Theorem~\ref{thm:div-U}
motivate the following definition, in analogy to the goal oriented (GO) divergence \eqref{eq:intro:godiv:2} defined for a single QoI $f$:

\begin{definition}[Concentration/Information Divergence]
Given the notation and assumptions of Theorem~\ref{thm:generalized-go-bounds} and Theorem~\ref{thm:div-U}, we define the concentration/information divergence between a baseline model $P$ and the family  of  models $\efam$, satisfying \eqref{eq:U-bounds-bias}
for all   QoIs  in $\Fp$:
\begin{equation}
    \Upm \eta\Fp:=\inf_{c>0}\left \{\frac{1}{c}\log\Phi(\pm c)+\frac{1}{c} \eta^2 \right\}\, , \label{eq:U-divpm:def}
  \end{equation}
where $\efam$  and $\Fp$, are defined in    \eqref{Qfamily}
and \eqref{eq:Fp} respectively.

\end{definition}

\begin{remark} [Features of Concentration/Information Inequalities]
{\rm 
While the GO divergence bounds \eqref{eq:intro:bias-bound-go:sec2} are defined for a specific QoI $f$,  key features  of the 
new bounds in Theorem~\ref{thm:generalized-go-bounds} include:  \begin{enumerate*} \item[(a)] allow to consider whole families of admissible QoIs
$\Fp$ defined in (\ref{eq:Fp}),
and \item[(b)] they bypass the costly MGF
calculations needed in the GO divergence  \eqref{eq:intro:godiv:2}. \end{enumerate*} 
 Finally,  we next show that the new bounds \eqref{eq:U-bounds-bias} still share the advantages of the  GO divergence bounds, namely: in    Section~\ref{sec:sharpness-and-robustness} we  prove that   \eqref{eq:U-bounds-bias} is, (c) tight in the family of models $\efam$, \eqref{Qfamily}, and the family of QoIs $\Fp$, \eqref{eq:Fp}.
 in Section~\ref{sec:scalability} we show that \eqref{eq:U-bounds-bias} is,  (d) scalable to   QoIs that depend on large numbers of data such as statistical estimators  and  to high dimensional probabilistic models.  
}
\end{remark}

We will next discuss   specific examples of the bound $\Phip{c}$ in the  concentration bounds  \eqref{eq:conc:mgf-bound:thm} and Theorem~\ref{thm:generalized-go-bounds}; furthermore, we also demonstrate  how we can select 
such  concentration bounds depending on the information we have regarding the distribution $P$. We divide our presentation into two  cases, namely bounded and unbounded QoIs $f$.

\subsection{Sub-Gaussian Bounds}
\label{sec:unbounded-obs}
For an unbounded QoI $f$ and a probability distribution $P$, we can characterize
the type of concentration by bounding either the tail probabilities $P(f(X)-\Ep[f]>a)$ for all $a$ or
$\Mp{c;\ti f}$ for all $c$ for which the MGF is finite. In this section, we discuss the (classical) sub-Gaussian 
concentration bounds which are characterized by Gaussian
decay of the tails.  
Sub-exponential bounds (see
Section~\ref{ex:expo-dist}) 
and sub-Poissonian bounds 
could also be useful in 
various situations but we will not discuss them further here (see e.g. \cite{ledoux2005concentration}).



\sepeq{Sub-Gaussian concentration bounds   ~\cite{concentration}
:} We say that $f=f(X)$ is a sub-Gaussian random variable if there
exists a $\sib>0$ such that
\begin{align}
  \label{eq:sub-Gaussianity}
  \Mp{c;\ti f}\leq \Phip{c}:=\exp(c^2\sib^2/2) \text{ for all } c\in\mathbb{R}.
\end{align}
Now given  a fixed $\sib$, we can consider the family of QoIs defined in \eqref{eq:Fp}, 
\begin{equation}
\label{eq:Fp:subG}
\Fp:=\{g:\Mp{c;\tilde g}\leq \Phip{c}=\exp(c^2\sib^2/2)\}\, ,
\end{equation}
i.e. we consider  all random variables with MGF bounded by the MGF of a normal random variable with
variance $\sigma_B^2$. Furthermore, using \eqref{eq:U-divpm} we can write an explicit formula for
$\Upm{\eta}\Fp=\inf_{c>0}\{\frac{c \sigma_B}{2}+ \frac{\eta^2}{c}\}$ as 
\begin{align}
  \label{eq:sub-G-Upm}
  \Upm \eta\Fp=\sib \sqrt{2} \eta. 
\end{align}
By expanding
$\Mp{c;\ti f}$ around $c=0$, we can readily show that $\sib^2$ is an upper bound
of $\Vp[f(X)]$. 
Relation~(\ref{eq:sub-G-Upm})  also implies that there is no
$\eta$-admissible model $Q\in \mathcal{Q}_\eta$ for which the QoIs under
consideration lie beyond the uncertainty  region  given by
Theorem~\ref{thm:generalized-go-bounds}:
\begin{equation}
  -\sib \sqrt{2}\eta\leq \E_Q[g]-\E_P[g] \leq
  \sib \sqrt{2}\eta
\end{equation}
for all models $Q\in \mathcal{Q}_\eta$ and QoIs $g \in
\Fp$. In Corollary~\ref{cor:coeff-of-var}, we consider the special case where $P$
is a normal distribution which is compared against any models $Q$--possibly not normal--from $\efam$. 

\begin{corollary}
  \label{cor:coeff-of-var}
  Consider the QoI $f(x)=x$ where $P=N(\mu,\sigma^2)$. Also, let $Q$ be any distribution such that $\R QP\le \eta^2$. Then, if the coefficient of variation (also known as relative standard deviation) is $c_v:=\sigma/|\mu|$, the relative model bias satisfies: 
  \begin{equation*}
    -c_v \sqrt{2}\eta \le \frac{\Eq[f]- \Ep[f]}{|\Ep[f]|} \leq c_v \sqrt{2}\eta \, . 
  \end{equation*}
\end{corollary}


In general, sub-Gaussianity is a strong assumption for an unbounded random
variable. For example, if $P=\mathrm{Laplace}(1)$, \ie a two-sided exponential
distribution centered at zero, then $\Mp{c;X}=1/(1-c^2)$, $|c|<1$, which cannot
be bounded by any $\exp(c^2 \sib^2/2)$ for all $c$. 
Finally, we note that results like the McDiarmid's inequality, see
Section~\ref{sec:scalability} below,  or  the logarithmic Sobolev inequalities~\cite{raginsky2013concentration,bobkov1999exponential},  can provide values for the constant 
$\sib^2$ for  QoIs that satisfy specific properties, \eg \eqref{eq:mcdiarmid-cond}.


\subsection{Bennett and Hoeffding Bounds}
\label{sec:bounded-obs}
Many quantities of interest are bounded such as failure  probabilities or   functions of random variables with bounded support.
Bounded random variables are necessarily sub-Gaussian~\cite{concentration},
but much sharper bounds for their MGFs, \eqref{eq:conc:mgf-bound:thm}, can be derived and used to bound the
worst-case bias through Theorem~\ref{thm:generalized-go-bounds}. In this direction, we next
discuss some additional  concentration bounds for bounded QoIs, that we will also showcase  in examples in this work. This list is not
complete by any means and other concentration inequalities can be used here, 
see for instance~\cite{raginsky2013concentration} for other bounds. For each case below, the family of
QoIs $\Fp$ is defined in terms of the concentration bound on the MGF, \eqref{eq:Fp}, as in Theorem~\ref{thm:generalized-go-bounds}.

\sepeq{Bennett concentration bound}~\cite[Lemma 2.4.1]{dembo2010large}: Consider the   random variable $X$ where $X\sim P$  and the  QoI $f=f(X)$ such that $f(X)\leq b$,  for some $0 \le b < \infty$. Setting
$\mu:=\Ep[f(X)]$, $\tilde{b}:=b-\mu$, we have
\begin{equation}
\label{eq:bennet-mgf-ineq}
\begin{aligned} 
  \Mp{c;\ti f}\leq \Phip{c}:=&\frac{\tilde{b}^2}{\tilde{b}^2+\sib^2}\exp(-c\sib^2/\tilde{b})
  +& \frac{\sib^2}{\tilde{b}^2+\sib^2}\exp(c\tilde{b}), 
\end{aligned}
\end{equation}
for all $c\geq 0$ and where $\sib^2$ is any upper bound of $\Vp[f]$. Therefore, keeping in mind Remark~\ref{rem:admiss}, we define
\begin{equation}
\label{eq:Fp:Bennett}
\begin{aligned}
\Fp=&\{g:\Mp{c;\tilde g}\leq \Phip{c}\, 
\}\, ,   &\mbox{where $\Phi$ is defined in \eqref{eq:bennet-mgf-ineq}.}
\end{aligned}
\end{equation}

\sepeq{Bennett-$(a,b)$ concentration bound}~\cite[Corollary 2.4.5]{dembo2010large}: If the QoI $f$ is
such that $a\leq f(X)\leq b$, $X\sim P$, then we can set $\sib^2=(\mu-a)(b-\mu)$
in the Bennett bound to obtain
\begin{equation}
\label{eq:bennet-(a,b)}
\begin{aligned}
  \Mp{c;\ti f}\leq \Phip{c}:=\frac{\tilde{b}}{b-a}\exp(c\tilde{a})-\frac{\ti {a}}{b-a}\exp(c\tilde{b})  \text{ for all }  c\in\mathbb{R}.
\end{aligned}
\end{equation}
The right-hand side of~\eqref{eq:bennet-(a,b)} is the MGF of a
Bernoulli-distributed random variable with values $\{a,b\}$. Note that the
Bernoulli is the distribution with the most ``spread'' around the mean value
between all bounded random variables in $[a,b]$. Similarly to \eqref{eq:Fp:Bennett} we have,
\begin{equation}
\label{eq:Fp:Bennettab}
\begin{aligned}
\Fp=&\{g:\Mp{c;\tilde g}\leq \Phip{c}\, 
\}\, , &\mbox{where $\Phi$ is defined in \eqref{eq:bennet-(a,b)}.}
\end{aligned}
\end{equation}

\sepeq{Hoeffding concentration bound}\cite{hoeffding1963probability,dembo2010large}: When the QoI $f$ is bounded as
in the Bennett-$(a,b)$ case, we can further bound the Bennett-$(a,b)$ bound  by a Gaussian MGF,
giving rise to  the (less tight)  Hoeffding MGF bound,
\begin{align}
  \label{eq:Hoeffding}
  \Mp{c;\ti f}\leq \Phip{c}:=\exp(c^2(b-a)^2/8) \text{ for all } c\in\mathbb{R}.
\end{align}
As in
the sub-Gaussian case of Section~\ref{sec:unbounded-obs}, we can calculate $\Upm
\eta\Fp$ explicitly:
\begin{align}
  \label{eq:hoef-u}
  \Upm \eta\Fp =  (b-a)\sqrt{2}\eta. 
\end{align}
Finally the set of QoIs is 
\begin{equation}
\label{eq:Fp:Hoeffding}
\begin{aligned}
\Fp=&\{g:\Mp{c;\tilde g}\leq \Phip{c}\, 
    \},  &\mbox{where $\Phi$ is defined in \eqref{eq:Hoeffding}.}
\end{aligned}
\end{equation}
\begin{table}
  \centering
  \caption{\label{tab:MGF-bounds} The different MGF bounds along with the
    conditions they impose on $P$ and $f$ and the quantities they depend on for
    their implementation if we are interested in quantifying the worst-case
    bias. However, bounding the worst-case $\Eq[f]$ \textit{does not} require
    $\Ep[f]$. Gaussian decay of the tails of the distribution of $f(X)$ implies
    the sub-Gaussian MGF bound (similar assumptions about the tails exist for
    the rest of the bounds). In terms of information requirements, the Hoeffding
    bound requires the least amount, but it is also the least tight. As
    available information/data  for the bounds grow, the bounds get tighter.}
  \begin{tabular}{|l|c|c|c|r|}
    \hline
    Name & Conditions on $f,P$ & $\Phi=\Phi(c)$ input \\
    \hline
    Hoeffding~(\ref{eq:Hoeffding}) & $a\leq f(X)\leq b$& $a,b$  \\
    Bennett-$(a,b)$~(\ref{eq:bennet-(a,b)})& $a\leq f(X) \leq b$& $\Ep[f],a,b$   \\
    Bennett~(\ref{eq:bennet-mgf-ineq})& $f(X)\leq b$, $\Vp[f]\leq \sib^2$& $\Ep[f],b,\sib$ \\
    \hline \hline
    sub-Gaussian~(\ref{eq:sub-Gaussianity})& $\Mp{c;\ti{f}}\leq \exp( \sib^2c^2/2)$ & $\sib$\\
    \hline \hline 
    GO bound~(\ref{eq:intro:godiv:2}) & $\Mp{c;\ti f} < \infty $ &  $\Ep[(f)^k]$ for all $k$\\
    \hline
  \end{tabular}
\end{table}

\begin{remark}[Hierarchy of bounds]{\rm
  \label{rem:order}
 It is straightforward to demonstrate that we can order
  the bounds in terms of accuracy, noting that  if the QoI $f$ is bounded in $[a,b]$,  then we always have the bound
  $\sib^2\leq (\Ep[f]-a)(b-\Ep[f])$ in the Bennett bound \eqref{eq:bennet-mgf-ineq}. Therefore, we have the hierarchy of concentration bounds:
  \begin{align}
    \label{eq:ordering}
    \Mp{c;\ti f} 
    \leq \text{Bennett} \leq \text{Bennett-(a,b)}\leq \text{Hoeffding}. 
  \end{align}
  Unlike the two Bennett bounds,  the Hoeffding bound is independent of the location of the mean $\mu$ within the
interval $[a,b]$ and only depends on the length of the interval $b-a$.  As such, it
requires the least amount of information about $f$ and $P$ and is the least
sharp of the bounds, see Table~\ref{tab:MGF-bounds}
and the requirements for the QoI families $\Fp$, \eqref{eq:Fp:Bennett}, \eqref{eq:Fp:Bennettab} and \eqref{eq:Fp:Hoeffding}.  On the other end, the GO divergence bound---involving $\Mp{c;\ti f}$---is the  tightest, as we see in \eqref{eq:ordering}, but also the most expensive to implement, see 
Table~\ref{tab:estimator-cost}. We  also refer to a demonstration of this hierarchy in the example in Section~\ref{ex:truncated-normal}. Overall, as
    available information/data  on the QoI $f$ and and the baseline model $P$ grows, concentration  bounds and therefore model bias bounds become tighter.
    Finally, we  refer to Figure~\ref{fig:contour-bennet}, where we demonstrate the tightness of the model bias bounds \eqref{eq:U-bounds-bias},
    \eqref{eq:U-divpm}, 
    in terms of both $\eta^2=\R QP$ and $\sib$, for the  
    Bennett bounds \eqref{eq:bennet-mgf-ineq}.
    }
\end{remark}

\begin{figure}
  \centering \includegraphics[width=0.499\textwidth]{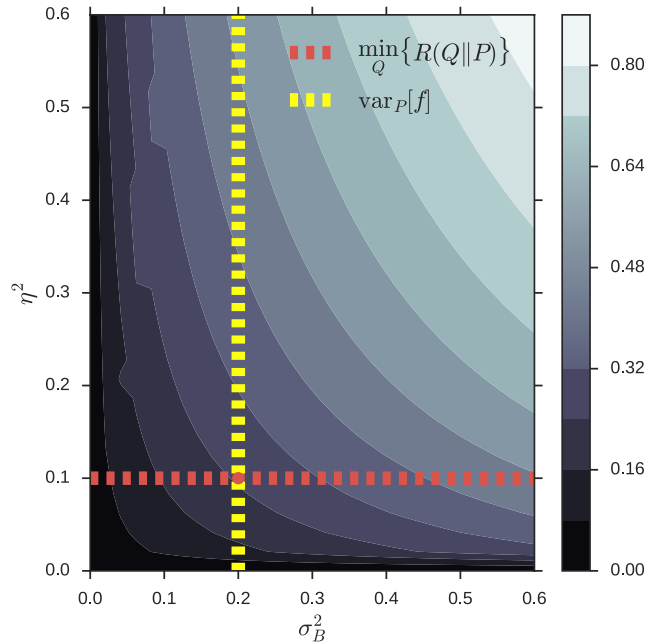}
  \caption{\label{fig:contour-bennet} Level curves of the upper  model bias bound \eqref{eq:U-bounds-bias} with the
    Bennett bound \eqref{eq:bennet-mgf-ineq} and assuming $b=1$, $\Vp[f]\leq \sib^2$. Knowing $\eta^2=\R QP$ for
    some model $Q$ and an upper bound on the variance provides model-bias
    guarantees (through Theorem~\ref{thm:generalized-go-bounds}). Further
    reduction of the model bias bound requires a corresponding---and potentially
    expensive---decrease in KL and/or a tighter upper-bound for $\Vp[f]$, for
    example, by incorporating additional data. The tightest possible guarantee
    afforded by the Bennett bound is gained when $\sib^2=\Vp[f]$ and
    $\eta^2=\min_Q{\R QP}$.}
\end{figure}

\begin{remark}{\rm 
\label{rem:size:Fp}[How large is the class 
$\Fp$?]
 A plausible question is how rich is the set of admissible QoIs, $\Fp$ derived by the various  concentration bounds on the MGF $\Mp{c;\ti g}$ in
\eqref{eq:Fp:subG}, \eqref{eq:Fp:Bennett}, \eqref{eq:Fp:Bennettab} and   \eqref{eq:Fp:Hoeffding}. Here we address this question in the context of the Bennett bound, however the same argument also applies to the Bennett-$(a,b)$ and Hoeffding bounds,
as well as to the sub-Gaussian case in Section~\ref{sec:unbounded-obs}.
We can get a simple first  insight in this direction based on \eqref{eq:bennet-mgf-ineq}. Indeed, based on the conditions for this inequality to hold, we readily have that 
$$
\Fp \supset \{g: \; g(X)\leq b\, , \;  \Vp[g] \le \sib^2\, , \; \Ep g=\mu\, 
\}\, .
$$
We also note that enforcing the condition on the mean, $\Ep g=\mu$, is trivial and involves only a translation of the QoI $g$. 
}
\end{remark}

\section{Tightness of the Concentration/Information Inequalities}
\label{sec:sharpness-and-robustness}
In this section we show that, under suitable assumptions, for the concentration/information bounds derived in Section~\ref{sec:conc-ineq} the divergence  $\Upm \eta \Fp$ retains some of the tightness properties of the  GO divergence  $\go QPf$  established in Section~\ref{sec:outline}.

\begin{theorem}
  \label{thm:sharpness-of-U}
Let $P$ be a probability  and  $ \mathcal{Q}_\eta=\{Q: R(Q\|P) \le \eta^2\}$. 
Assume $\Phi(c)=M_{P}(c; \tilde f)$ is   a
MGF for some QoI $f$ with respect to $P$ and let 
\begin{equation}
    \Fp=\{g:\Mp{c;\tilde g}\leq \Phip{c}\, \textrm{ for all } c  \in \mathbb{R}\}\,.
  \end{equation} 
Then,  there exist 
probabilities $\Ptil{\pm} \in \mathcal{Q}_\eta$ (see  \eqref{eq:parametric}) that satisfy $\R {\Ptil{\pm}}P=\eta^2$ 
and 
\begin{equation}
 \label{eq:U-sharp-plus}
\begin{aligned}
 \Up \eta\Fp &= \E_{\Ptil{+}}[f]-\E_{P}[f]
 &=\max_{Q\in \mathcal{Q}_\eta \,, g \in \Fp}\E_{Q}[g]-\E_{P}[g],
\end{aligned}
\end{equation}
\begin{equation}
 \label{eq:U-sharp-minus}
 \begin{aligned}
    -\Um \eta\Fp & = \E_{\Ptil{-}}[f]-\E_{P}[f]  &=
                   \min_{Q\in \mathcal{Q}_\eta\,, g \in \Fp}\E_{Q}[g]-\E_{P}[g]\, ,                  
 \end{aligned}  
\end{equation}
\ie the maximum and minimum for 
model bias is attained within the family of models $\mathcal{Q}_\eta$ and the family of QoIs $\Fp$, see the schematic in Figure \ref{fig:tightness}. 
We also have the  ``confidence band" around the baseline model $P$,

  \begin{equation}
    \label{eq:thm:tightness:main}
    \begin{aligned}
    -\Um \eta\Fp +\E_{P}[g]
    \le
    &\E_{Q}[g] \le \E_{P}[g]+
    \Up \eta\Fp
    \, &\mbox{for all $Q \in \mathcal{Q}_\eta$,  $g \in \Fp$\, ,}   
    \end{aligned}
  \end{equation}
  with the two equalities holding if $Q=P^{c_\mp}$ respectively and for $g=f \in  \Fp$.
\end{theorem}

  

\begin{proof} Since $f\in \Fp$, Theorem \ref{thm:tightness} implies that the probabilities 
 $\Ptil{\pm}$ in \eqref{eq:parametric}, with 
$c_\pm$ chosen such that  
$\R
  {\Ptil{\pm}}P=\eta^2$ satisfy
  \begin{align}
    \label{eq:U-equals-GO}
    \go {P^{c_\pm}} P {\pm f} = \Upm \eta\Fp \,.
  \end{align}
 Therefore, by 
 Theorem~\ref{thm:generalized-go-bounds}, \eqref{eq:U-bounds-bias}, for all $Q \in \mathcal{Q}_\eta, g \in \Fp$
  \begin{equation}
    \label{eq:U-bounds-bias-tight}
    -\Xi(P^{c_-}||P; -f) \leq \Eq[g]-\Ep[g] \leq \Xi(P^{c_+}||P; f)\, .  
  \end{equation}
  Finally, we apply \eqref{eq:max} and \eqref{eq:min} of
  Theorem~\ref{thm:tightness} and 
  use~\eqref{eq:U-equals-GO} to
  conclude the proof.
\end{proof}

 \begin{figure}[!t]
    \centering \includegraphics[width=0.5\textwidth]{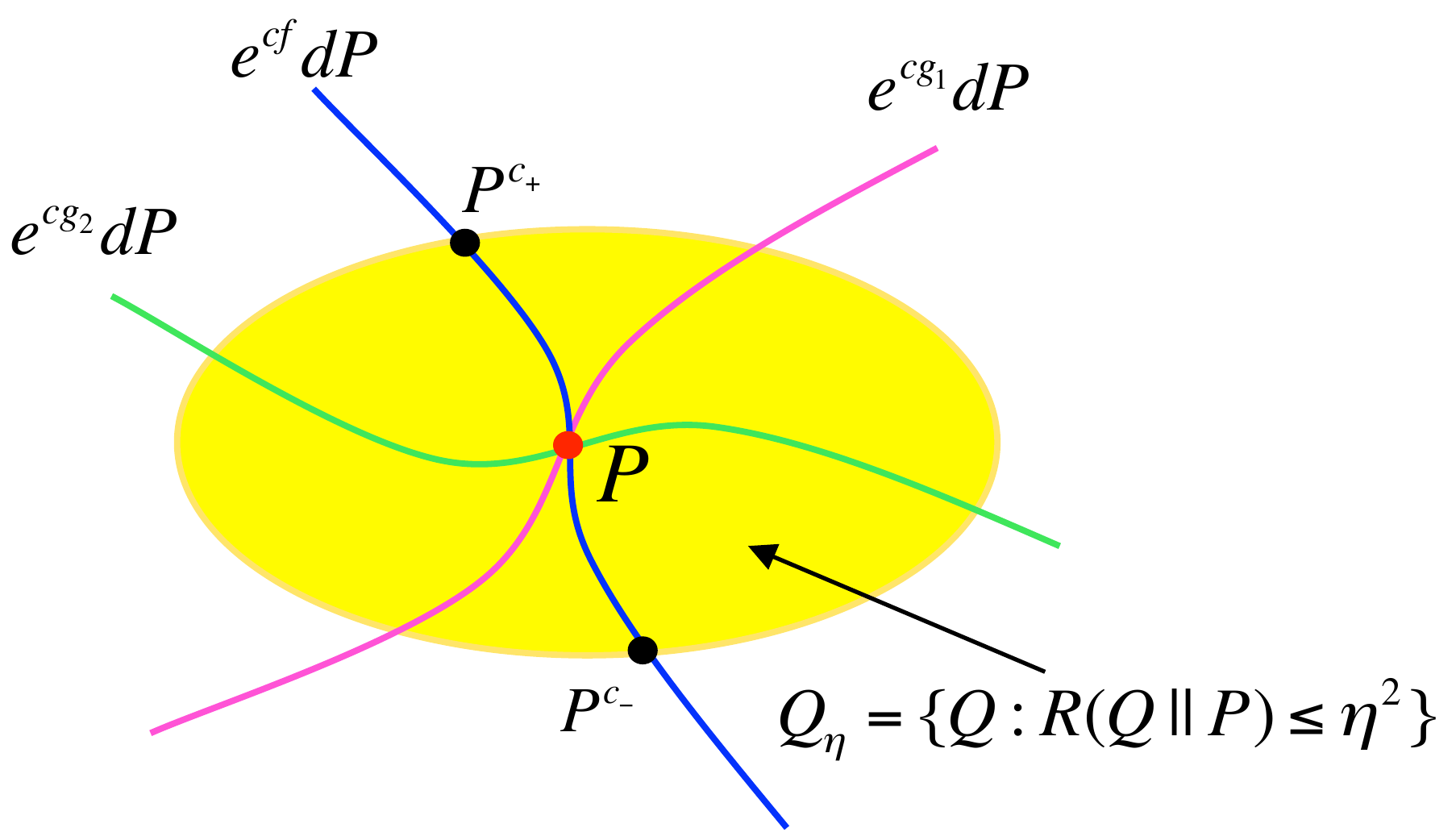}
    \caption{ The schematic depiction of Theorem~\ref{thm:sharpness-of-U} for a family of
     QoIs  $\Fp$ and tolerance $\eta^2$. The solid lines depict the one-parameter
      tilted probability distributions $P^c$ corresponding to the QoI $g_1, g_2, f
      \in \mathcal{F}_P$. The theorem implies that the upper and lower bounds in
      the family $\mathcal{Q}_\eta=\{Q: R(Q\|P) \le \eta^2\}$ are attained at
      the probability measures $Q^{\pm}=P^{c_{\pm}}$. }
    \label{fig:tightness}
  \end{figure}

\begin{remark}[Connections to Mass Transport]
{\rm
\label{rem:mass:transport}
The proof of Theorem \ref{thm:sharpness-of-U} is quite straightforward and we discuss here one approach to verify the crucial assumption of the Theorem, namely that 
\begin{equation}
 \label{eq:tightness:real:assumption}
 \Phi(c)=M_{P}(c; \tilde f)\, \quad \mbox{for some QoI $f$\, .}
 \end{equation}
One natural way to ensure this is intimately 
related to mass transport methods, \cite{villani2008optimal}. Instead of 
\eqref{eq:tightness:real:assumption} we may assume the more easily checkable hypothesis
  that $\Phi(c)=M_{\bar P}(c; \tilde h)$ 
for some $h$ and some model $\bar P$;  e.g. $h(x)=x$ and $\bar{P}$ a 
Gaussian distribution for the Hoeffding's bound, see also  Example~\ref{ex:sub-g-tight}
 below. 
To prove \eqref{eq:tightness:real:assumption} one  shows then that there exists  a {\it  transport map} between $P$ and $\bar P$, namely a map $T$ such that $\bar P(A)= P(T^{-1}(A))$ for any measurable set $A$ \cite{villani2008optimal}. If a transport map exists we have 
 $$
  \mathbb{E}_P[h \circ T]=\int h(Tx)P(dx)=\int h(y)\bar P(dy)=\mathbb{E}_{\bar
    P}[h]\,.
  $$
and hence with $f=h\circ T$ 
 $$
  \tilde f=f-\mathbb{E}_P[f]=h\circ T-\mathbb{E}_{\bar P}[h]=\tilde
  h \circ T\, .
  $$ 
This implies that 
\begin{align*}
\Phi(c)= M_{\bar P}(c; \tilde h)=\int e^{c{\tilde h}(y)}\bar P(dy) = \int
  e^{c{\tilde h} (Tx)} P(dx)= M_{P}(c; \tilde{f})\, ,
\end{align*}
 and thus the assumption \eqref{eq:tightness:real:assumption} holds.

  
To ensure the existence of such a transport map $T$ 
one needs some assumptions on $P$ (and $\bar P$).  
For example, if $P$ and $\bar P$ are non-atomic 
measures then a transport map always exists. 
If the measure $P$ has a density then $T$ can be constructed using the Knothe-Rosenblatt rearrangement or Brenier's $L_2$ optimal transport  map; we refer to  Chapter 1~\cite{villani2008optimal,Carlier2010}  for more details on these maps, and several  other such transport maps and relevant conditions for their existence.
}
\end{remark}

Next we demonstrate how to use Theorem~\ref{thm:sharpness-of-U} by interpreting
$\Phip{c}$ as the MGF of a suitable QoI $f$ with respect to the distribution $P$. 
In Example~\ref{ex:bennet-bernoulli}, we illustrate the tightness of the
concentration bounds for the case of bounded random variables supported in $[-1,1]$, while   
the arguments can be trivially generalized to any other
bounded interval.

\begin{example}[Bennett-(a,b) QoIs]
{\rm 
  \label{ex:bennet-bernoulli}
  Consider a distribution $P$ such that there is  an event $A\subset \mathbb{R}$  such that  $P(A)=1/2$; we consider 
the  family of QoIs, $\Fp$, for which~(\ref{eq:bennet-(a,b)}) is true with $a=-1$, $b=1$, $\Ep[g]=0$ for
  all $g\in \Fp$. The corresponding Bennett-(a,b) bound is
  \begin{align*}
    \Phip{c}=\frac{1}{2}e^c+\frac{1}{2}e^{-c} \, . 
  \end{align*}
Then, if 
we choose 
   $ f(x):=2\cdot 1_{A}(x)-1$, 
  where $1_A$ is the characteristic function of the set $A$, we have $\Phi(c) =
  M_{P}(c; \tilde f)$.
  Therefore 
Theorem~\ref{thm:sharpness-of-U} is immediately applicable. 
}
 \end{example}
The next example  
covers the case of sub-Gaussian QoIs which contains both bounded and
unbounded random variables.
\begin{example}[sub-Gaussian QoIs]{\rm 
  \label{ex:sub-g-tight}
Consider the probability measure  $P$ on $\mathbb{R}$ which has a density. 
For sub-gaussian QoIs \eqref{eq:sub-Gaussianity}
we have the bound $\Phip{c}=\exp(c^2\sib^2/2)$, however, we  can rewrite the bound as 
$$
\Phip{c}=M_{\bar P}(c; \tilde h)
$$
where  $h(x)=x$ and $\bar P=N(0,\sib^2)$ is a normal 
distribution.  Since $P$ has a density, we can use the measurable isomorphism, or any other applicable map discussed in  Remark~\ref{rem:mass:transport}, to construct a transport map $T$ between $P$ and $\bar P$. 
Thus, we can  show the existence of a a QoI $f$ that satisfies the condition \eqref{eq:tightness:real:assumption}
and we can readily
apply 
Theorem~\ref{thm:sharpness-of-U} to show the tightness of the  bounds given by \eqref{eq:sub-G-Upm}. 
 }
\end{example}

\section{Model Bias for Statistical Estimators}
\label{sec:scalability}

As discussed in Section \ref{sec:outline} a key challenge is 
to  control the risk involved in evaluating statistical estimator
using the baseline model $P$ rather than the true model $Q$. 
In addition it is important to control the bias of QoI
which are not necessarily  expected values, for example the 
bias in the variance, i.e, $\Vp{X}-\Vq{X}$, or other statistics such 
as correlation of skewness, or mean and quantiles, see~\cite{wasserman2004}. 
Generally, given data $\nn X$, we aim to control the bias of statistical 
estimator $\psi=\psi(\nn X)$, for example the sample 
variance~\eqref{eq:sample-var}. 


To obtain useful bounds on the bias of  statistical estimators 
$\psi$, we need to exhibit and control 
the dependence of the inequalities
in Sections~\ref{sec:bounded-obs} and~\ref{sec:unbounded-obs} on the amount of data available, i.e. the dependence on $n$. We 
will exhibit  a large and natural class of statistical estimators 
for which inequalities asymptotically independent on $n$. As 
demonstrated in \cite{ katsoulakisjcp2017} the Concentration/Information inequalities  of Sections \ref{sec:outline} and 
\ref{sec:conc-ineq} are the only known information equalities which 
scale properly with $n$.

The main tool we shall use is the key result used in the  proof 
of the McDiarmid's inequality, see also the Hoeffding-Azuma bound, \cite{dembo2010large}. We refer to Chapter 2 of  \cite{raginsky2013concentration} or ~\cite{mcdiarmid1989method} for the proof.

\begin{proposition} \label{thm:McDiarmid}
Let $\nn X$ be independent random variables with joint distribution $P^n=P_1\times\cdots\times P_n$. 
Let $\psi(\nn x)$ satisfy the Lipschitz condition
%
%
\begin{equation}
\label{eq:mcdiarmid-cond}
\begin{aligned}
    &\sup_{\nn x,x_k'}|\psi(x_1,\ldots,x_k,\ldots, x_n)-\psi(x_1,\ldots,x_k',\ldots, x_n)| 
    &\leq d_k \, .
  \end{aligned}
\end{equation}
for some constants $d_k$, $k=1,\ldots,n$. Then $\psi(\nn X)$ is 
a sub-Gaussian random variable and for all  $c \in \mathbb{R}$ 
we have
\begin{equation}
 \label{eq:mcdiarmid-res}
 \begin{aligned} 
    M_{P^n}(c; \tilde \psi)&=\E_{P^n}\left [\exp\left(c(\psi-\E_{P^n}[\psi])\right)\right]  &\leq \exp \left( \frac{c^2}{8}\sum_{k=1}^{n}d_k^2\right) \, . 
  \end{aligned} 
\end{equation}
\end{proposition}

By combining the bound in (\ref{eq:mcdiarmid-res})
with the definition of $\Upm \eta\Fp$ in 
Theorem~\ref{thm:generalized-go-bounds} for the sub-Gaussian 
case \eqref{eq:sub-Gaussianity} we obtain immediately

\begin{theorem}
\label{thm:McDiarmidConc}
For  $\nn X$ and  $\psi(\nn x)$ 
as in Proposition \ref{thm:McDiarmid} we have 
\begin{equation}
\begin{aligned}
&\left|\E_{P^n}[\psi(\nn X)]  - \E_{Q^n}[\psi(\nn X)]\right|  = & \left( \sum_{k=1}^{n}d_k^2\right )^{1/2} \sqrt{2 \sum_{k=1}^n\R{Q_i}{P_i}}\,.
\end{aligned}
\end{equation}
If $\nn X$ are identically distributed with  common distribution $P$ 
and if there exists a constant $C$ such that 
\begin{align*}
d_k\leq \frac{C}{n} \,,\quad k=1, \ldots, n 
\end{align*}
then we have for any $n$
\begin{equation}
\label{eq:upm-mcd-prod}
\begin{aligned}  
  &\left|\E_{P^n}[\psi(\nn X)]  - \E_{Q^n}[\psi(\nn X)]\right| 
  \le  &C\sqrt{2 \R{Q}{P}}.
\end{aligned}
\end{equation}
\end{theorem}

Next, we apply  these results towards obtaining model bias bounds for statistical estimators.  

\sepeq{CDF estimator:} If $X$ is a real-valued with cumulative distribution function (CDF) $F_P(x)=P\{X \le x\}$ then given 
i.i.d. data $\nn X$  
\begin{align}
{\hat F}_n(x) = \frac{1}{n}\sum_{k=1}^n I_{\{X_k \le x\}} \,,
\end{align}
where $I_A$ is the indicator function of the set $A$, is an 
estimator for the CDF $F_P=F_P(x)$. It is easily verified that the 
conditions of Theorem \ref{thm:McDiarmidConc} are satisfied with 
$C=1$. Since the bound is uniform in $x$, and ${\hat F}_n(x)$ is an unbiased estimator of $F_P(x)$, we obtain 
\begin{equation}
\label{eq:CDFest}
\begin{aligned}
 &{\sup_{x} \left| F_Q(x)-F_P(x)\right|}
 =&\sup_{x} \left| \E_{Q^n}\left[{\hat F}_n(x)\right]-\E_{P^n}\left[{\hat F}_n(x)\right]\right| 
 \le &\sqrt{2 \R QP} \, ,
\end{aligned}
\end{equation}
for any alternative model $Q$ to the baseline $P$. As we also note in the sample variance example below, the estimator does not need to be unbiased. 


\sepeq{Sample variance and general statistical estimators:} The McDiarmid's inequality and the
condition \eqref{eq:mcdiarmid-cond} can be used to control bias 
of QoIs which are not simply expected values, for example the sample
variance 
\begin{equation}
\begin{aligned}
 V_n(\nn{X})&=\frac{1}{n-1}\sum_{i=1}^{n}\left(X_i-\frac{1}{n}\sum_{j=1}^n X_j\right)^2
 &=\frac{1}{2n(n-1)}\sum_{i,j=1}^{n}(X_i-X_j)^2.
\end{aligned}
 \label{eq:sample-var}
\end{equation}
If we assume that $|X_i|\le M$ for some $M>0$ then we have 
\begin{align*}
 \sup_{\substack{|x_i|\le M, \\ |x_k'|\le M}} |V_n(x_1,\ldots,x_k,\ldots,x_n)-V_n(x_1,\ldots,x_k',\ldots,x_n)|
 \leq  \frac{8M^2}{n-1} \,. 
 \end{align*}
Then the sample variance satisfies \eqref{eq:mcdiarmid-cond} with $d_k=8M^2/(n-1)$ for all $k$. Thus, we can bound the corresponding model bias by
\begin{equation}
  \label{eq:model-bias-sample-variance}
  \begin{aligned}
    |\Vp[X]-\Vq[X]|&=|\E_{P^n}[V_n]-\E_{Q^n}[V_n]| 
    &\leq 8M^2 \frac{n}{n-1}\sqrt{2\R{Q}{P}}.
  \end{aligned}
\end{equation}
valid for all $n>1$. Note that if take $n \to \infty$ we obtain 
the variance bound 
\begin{align*}
 |\Vp[X]-\Vq[X]|\leq 8M^2 \sqrt{2 \R QP}
\end{align*}
which shows how KL-divergence $\R{Q}{P}$ control the misspecification 
for QoIs beyond expected values.  The same analysis also applies to the (biased) plug-in estimator for the variance, namely
$$
\tilde V_n(\nn{X}):=\mbox{var}_{{\hat F}_n}[X]=\frac{1}{n}\sum_{i=1}^{n}\left(X_i-\frac{1}{n}\sum_{j=1}^n X_j\right)^2\, .
$$

Finally, we can easily generalize the sample variance calculation to more 
general QoIs and statistical estimators. The sample variance depends (up to a factor $\frac{n-1}{n}$) only on the two sample averages  $\frac{1}{n}\sum_{i=1}^n X_i$ and $\frac{1}{n}\sum_{i=1}^n X_i^2$. It is not difficult to 
see that if $|X_i|\le M$ and the QoI has the form 
\begin{equation}
\label{eq:plugin:general}
\begin{aligned}
\psi_n(\nn X) = g\left(\frac{1}{n}\sum_{i=1}^n f_1(X_i), \ldots,   \frac{1}{n}\sum_{i=1}^n f_k(X_i) \right)
\end{aligned}
\end{equation}
for some $f_1, \cdots f_k$ (say the the first $k$ moments),   
and for some Lipschitz continuous function $g$, then one can apply Theorem \ref{thm:McDiarmidConc} for a constant $C$ which depends on $M$, the Lipschitz constant for $g$ and $f_1, \cdots f_k$.  One important example 
of the type \eqref{eq:plugin:general} is the sample correlation, we refer to Example 2.16 in \cite{wasserman2010}.

\sepeq{Confidence Bands and 
Model Bias}
To further  illustrate our results we construct 
a non-parametric confidence band for the CDF
$F_Q(x)$.  We combine the bound 
\eqref{eq:CDFest} with the Dvoretzky-Kiefer-Wolfowitz (DKW) inequality \cite{wasserman2004,wasserman2010}, 
i.e.  the  bound  
\begin{equation}
\label{eq:DKW:ineq}
P \left\{ \sup_{x}|{\hat F}_n(x) - F_P(x)|\ge \epsilon\right\} \le 2 e^{-2 n \epsilon^2} \, , 
\end{equation}
which itself is obtained though concentration inequalities.
For any $n$ and $\alpha > 0$, we set $\epsilon_n= \sqrt{\log(2/\alpha)/2n}$
and 
\begin{equation}
\begin{aligned}
L_n(x; \eta) = \max\{ {\hat F}_n(x) - \sqrt{2}\eta - \epsilon_n \,,\, 0\}\, \\ 
U_n(x; \eta) = \min\{ {\hat F}_n(x) + \sqrt{2}\eta + \epsilon_n \,,\, 1\} \, .
\end{aligned}
\end{equation}
Since ${\hat F}_n(x)$ is an unbiased estimator for the baseline model $P$ 
rather than for the (unknown) ``true" model $Q$
we obtain the $\alpha$--confidence band for $F_Q(x)$:
\begin{equation}
\label{eq:confidence:bias}
\begin{aligned}
P \left\{ L_n(x;\eta) \le F_Q(x) \le U_n(x;\eta) \textrm{ for all } x \right\}  \ge   1- \alpha \,,\quad \mbox{for all} \quad Q \in \mathcal{Q}_\eta\, .
\end{aligned}
\end{equation}
Due to the fact that both our bound \eqref{eq:CDFest} and the DKW 
inequality \eqref{eq:DKW:ineq} are valid for any data size $n$, the confidence band \eqref{eq:confidence:bias} does not require any asymptotic normality assumptions or a large data set  $n \gg 1$.

\sepeq{Connections to the Vapnik-Chervonenkis inequality}
The DKW inequality is an effective  tool for controlling deviations from the average  for one dimensional distributions  and their corresponding CDFs. However, the Vapnik-Chervonenkis (VC)
theory~\cite{wasserman2010} allows us to address the same issues in a more general setting  that is applicable to higher-dimensional distributions,  by considering the empirical probability distribution instead of the CDF. In particular,  corresponding inequalities to    \eqref{eq:DKW:ineq}, but for the empirical probability distribution, can be derived based on the  VC theory, see for instance Theorem 2.41 and Theorem 2.43 in~\cite{wasserman2010}. In turn the VC inequalities, along with our concentration information bounds \eqref{eq:upm-mcd-prod}  can allow us to obtain 
confidence intervals for  higher dimensional  distributions, similarly to   \eqref{eq:confidence:bias}.

\begin{remark}
\label{rem:poor:scalab}
[Poor scalability of certain information inequalities]
{\rm
A notable feature of the concentration/information inequalities is that they scale independently of  the number of data/random variables $n$, at least for classes of QoIs that satisfy \eqref{eq:mcdiarmid-cond}, as   demonstrated  in Theorem~\ref{thm:McDiarmidConc} and the subsequent examples. Furthermore, the bias bound \eqref{eq:upm-mcd-prod} remains  discriminating even if  $n \to \infty$.
The same scaling features are also shared with the GO divergence bounds \eqref{eq:intro:bias-bound-go:sec2}, see \cite{katsoulakisjcp2017}. 
On the other hand, classical information inequalities scale poorly with $n$. For example, in the case of the Pinsker inequality \cite{Cover2006ElementsofInformationTheory,tsybakov},  let us consider  the QoI (estimator)  \eqref{eq:plugin:general} for the i.i.d. random variables $\nn X$,  
$$
\psi_n(\nn X) = \frac{1}{n}\sum_{i=1}^n f(X_i)\, .
$$
Then, the Pinsker inequality becomes
\begin{equation}
\label{eq:Pinsker}
\begin{aligned}
& \left|\E_{P^n}[\psi(\nn X)]  - \E_{Q^n}[\psi(\nn X)]\right|  
 \le  & \| f \|_{\infty}\sqrt{2  R(Q^n \mid \mid P^n)}=O(\sqrt{n})\, ,  
\end{aligned}
\end{equation}
where we used that
$\| \psi \|_{\infty}=\| f \|_{\infty}$, and  
$
R(Q^n \mid \mid P^n)= n R(Q \mid \mid P)
$.
Therefore the Pinsker bound \eqref{eq:Pinsker} blows up
as $n>>1$, in contrast to the concentration/information inequality \eqref{eq:upm-mcd-prod} that remains discriminating and informative for any $n$.  
Other  model bias bounds  based on the  Renyi or  $\chi^2$ divergences (the latter known as the Chapman-Robbins inequality) or the Hellinger metric,   also scale poorly with   the size of data set   and/or with the number of variables $n$; we refer to  Sections~2.2--2.3 in \cite{katsoulakisjcp2017} for a complete discussion.
}
\end{remark}

\section{Elementary examples}
\label{ex:simple-examples}
Prior to discussing  applications involving more complex models in Section~\ref{sec:epistemic-UQ-conc}, here we demonstrate  the concentration/information inequalities we developed earlier to two  elementary  examples that allow easy analytic and computational implementations.

\subsection{Exponential distribution}
\label{ex:expo-dist}
We first consider the model bias  bounds using the GO divergence in  Theorem~\ref{thm:GO-bounds}, contrasted to the concentration/information divergence in Theorem~\ref{thm:generalized-go-bounds}.
In our first example,  the baseline model $P$ is  an exponential distribution. The
models $Q$ can be any distributions which are  absolutely continuous with respect
to $P$, hence $\R QP < \infty$.
Let $P$ be an exponential distribution with parameter
$\lambda_P=1$. The QoI is $f(X)=X$. The MGF of $P$ is $\Mp{c;X}=1/(1-c)$ and
thus it is finite in $(0,1)$, while otherwise it is infinite. Next, we let  $\eta$  be a model  uncertainty threshold and $Q$ any
distribution, not necessarily exponential or in any parametric family,  such that $\R QP\leq \eta^2$. 
%
We note that the distribution $P$ exhibits sub-exponential behavior, namely 
\begin{align}
  \label{eq:sub-expo}
  \begin{aligned}
  \Mp{c;f}=1+c+\frac{c^2}{1-c}  
  \leq 1+c+2c^2  
  \leq \exp(c+c^2/(2\sib^2)) :=\Phip{c}\, , \quad c \in (-0.5,0.5)\, ,
  \end{aligned}
\end{align}
where $\sib=1/2$ and the interval $(-0.5,0.5)$ is selected so that the bounds remain finite.  In general,  if we different  information on the location of
$\lambda_P$, \eg from data, then we can adjust the interval that $c$ lies in
accordingly. Here, the  concentration/information bound \eqref{eq:U-divpm} is then adjusted according to \eqref{eq:sub-expo},  using Theorem~\ref{thm:generalized-go-bounds} and 
the general concentration bound
\eqref{eq:conc:mgf-bound:thm}.
Although the MGF is known in this particular example, the use of the
concentration bound \eqref{eq:sub-expo} allows us to quantify the worst-case model bias for all QoIs
$g \in \Fp$, where
\begin{equation}
\label{eq:Fp-expo}
\begin{aligned}
  \Fp=\{g: \Mp{c;\ti g}\leq \Phi(c), \
  c\in (-0.5,0.5),\ \sib^2\le 1/4\}.
\end{aligned}
\end{equation}
Figure~\ref{fig:expo-dist} is a comparison of the GO-divergence and the concentration/information bound based on 
\eqref{eq:sub-expo},
 along with the
exact model bias for the case that $Q$ is also an exponential distribution with $\R
QP\leq \eta^2$ and $\eta\in [0,1.6]$.
\begin{figure}[!t]
  \centering
  \includegraphics[width=0.499\textwidth]{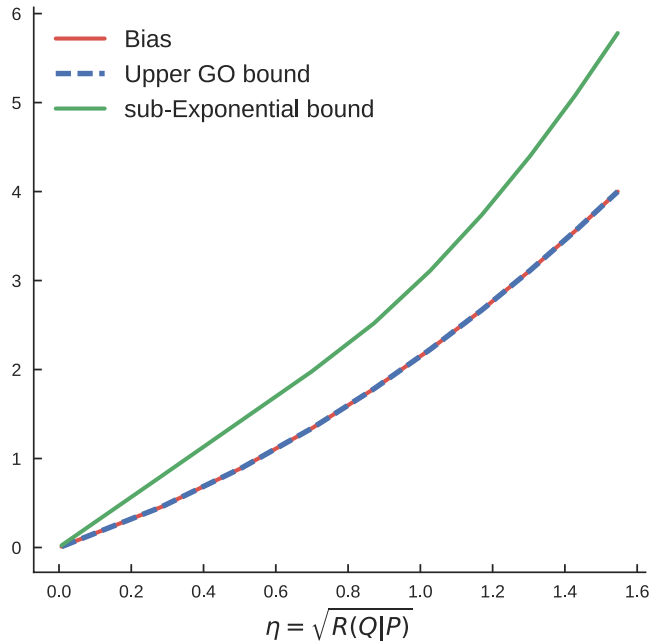}
  \caption{\label{fig:expo-dist} Comparison of model bias bounds
  based on   the GO divergence and the concentration/information 
\eqref{eq:sub-expo},
   with the exact model bias
    $1-\lambda_Q$, $\lambda_Q\in (1.01,10)$. The sub-exponential bound~(\ref{eq:sub-expo}) is less
    sharp as the KL divergence  increases, since  it  captures the worst-case
    performance over the family of QoIs $\Fp$. Although $\R QP$ is computed with
    $Q$ being an exponential distribution, the bounds to the  model bias are valid for
    any $Q$ that is absolutely continuous with respect to $P$ and has $\R QP$ in
    the range of the figure. }
  \label{fig:expo-comparison}
\end{figure}

Finally, we can consider other types of tail decay, and thus corresponding concentration inequalities,  besides the sub-Gaussian and the sub-exponential cases discussed thus far. 
For example, we can also consider   Poisson-type tail decay, see for instance  
\cite[Section 3.3.5]{raginsky2013concentration} 
and~\cite{boucheron2009concentration}.

\begin{remark}
{\rm
  The bias is an unbounded function of the KL divergence in this example---a
  consequence of the QoI $f(X)=X$ being unbounded under $P$. Therefore, any
  decrease in KL divergence translates to an improvement in worst-case model bias,
  see Figure~\ref{fig:expo-dist}; this fact is in
  sharp contrast with the truncated Normal example in 
  Section~\ref{ex:truncated-normal}, where even large   improvements to larger values of
  the KL divergence may not help much in reducing model bias, see Figure~\ref{ex:truncated-normal}.
}
\end{remark}

\subsection{Truncated Normal}
\label{ex:truncated-normal}

In this example the distributions we consider are bounded, allowing us to deploy the hierarchy of concentration/information bounds \eqref{eq:ordering}
developed in Section~\ref{sec:bounded-obs}.
We assume the random variable $X$ follows the truncated Normal distribution, $P=TN(0,1, -1, 1)$,
where $[-1,1]$ is the interval of support.  Here we will bound the model bias, $\bias$, for any $Q$
such that $\R QP=\eta^2$, where $\eta\in [0.01,1]$ and for any $f$ in a suitable family
of QoIs, $\Fp$. Apart from these, the bounds make no other assumptions on $Q$, $f$.
Figure~\ref{fig:compare-trunc-normal-bounds} contains a comparison of the
different concentration/information bounds \eqref{eq:ordering} from Section~\ref{sec:bounded-obs}. 
\begin{figure}[!t]
  \centering
  \includegraphics[width=0.5\textwidth]{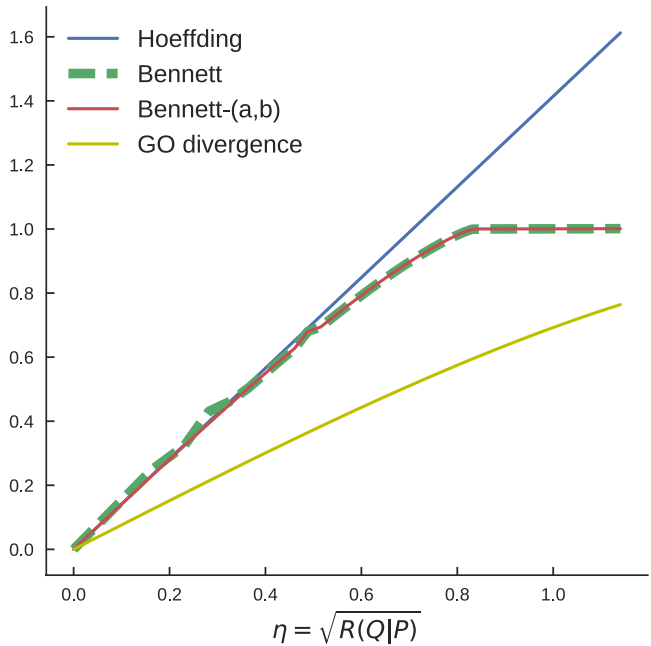}
  \caption{Comparison of the different bounds for the bias in the truncated
    Normal example (see Section~\ref{ex:truncated-normal}), assuming that the
    observable of interest is $f(X)=X$. This plot makes no assumptions on the
    form of $Q$ except that $\R QP=\eta^2\in (0.0,4.0)$.
    As in Figure~\ref{fig:expo-comparison}, here  the concentration/information bounds capture the worst-case
    performance over the family of QoIs $\Fp$, hence perform worse than the GO divergence bounds which are suitable only for a single QoI, see also \eqref{eq:ordering}.
    Notice that Bennett and
    Bennett-$(a,b)$  track  better the bound of the GO divergence for large
    values of the KL whereas the Hoeffding is  sufficient  only for small values of
    the KL, \ie at the linearized regime of the GO bounds. Only the upper bounds
    for the bias are shown here.}
  \label{fig:compare-trunc-normal-bounds}
\end{figure}

As a general observation, we notice that for large values $\eta=\sqrt{\R QP}$, small perturbations of
$\eta$ will not change the Bennett/GO (see Relations~(\ref{eq:intro:godiv:2})
and~(\ref{eq:bennet-mgf-ineq})) bounds significantly. Therefore, for some
QoIs, \eg $f(X)=X$, small improvements to large values of the KL will barely
improve the worst-case bias (as captured by the bounds, see
Figure~\ref{fig:compare-trunc-normal-bounds}). The existence of such QoIs is
guaranteed by the sharpness of the bounds demonstrated in
Section~\ref{sec:sharpness-and-robustness}. Finally, we not that even for the tighter concentration/information 
bounds, \ie the ones associated with the  two Bennett bounds  (\ref{eq:bennet-mgf-ineq}) and 
(\ref{eq:bennet-(a,b)}), there is some discrepancy with the GO divergence bound. This discrepancy is due to the fact that   the GO bound is applied only for a specific QoI, while the
concentration/information  bounds are tight over the broad classes of QoIs defined in  Section~\ref{sec:bounded-obs},
see also Remark~\ref{rem:size:Fp}.
 
 


\section{Epistemic Uncertainty Quantification via Concentration/Information Inequalities}
\label{sec:epistemic-UQ-conc}
In this Section, we apply the concentration/information inequalities to control model bias between baseline and alternative models in two
more complex examples. The  type of model bias considered here arises  in epistemic uncertainty quantification, where modelers are unsure if their baseline model included all necessary complexity or lacks sufficient data, \cite{Smith, Sullivan2015}. The KL divergence and in particular the GO divergence bounds provide a non-parametric framework to mathematically describe this type of epistemic uncertainties, as first shown in \cite{dupuis2011uq}.
Here, we consider two such examples that illustrate different aspects of epistemic uncertainty, namely a data-driven model for the lifetime of lithium batteries, as well as  a high-dimensional Markov Random Field model  subject to various localized uncertainties  such as local defects. 
A key aspect of our discussion in both examples 
is the necessity and the (ease of) implementation of concentration/information model bias bounds,
see for instance Remark~\ref{battery_remark}.

\subsection{Epistemic Uncertainty for Failure Probabilities}
\label{ex:battery-example}

Here, we apply the   bounds of Theorem~\ref{thm:generalized-go-bounds}, and in particular
the inequalities in Section~\ref{sec:bounded-obs}, to the life-time  analysis of lithium secondary
batteries. Firstly, we introduce the Weibull distribution which is widely used
in for analyzing life-time data, see \cite{eom2007life} and references therein. The  probability density function of a Weibull random
variable is
\begin{equation}
  f(t)=\frac{\beta}{\xi}\left (\frac{t}{\xi}\right)^{\beta-1}e^{-(\frac{t}{\xi})^\beta}, t>0,
\end{equation}
where $\beta > 0$ is called a shape parameter and $\xi > 0$ is a scale parameter of the
distribution \cite{weibulldef}. The shape parameter explains the types of failure and the scale
parameter explains the characteristic life cycle of devices. The cumulative
distribution function $F$ can be expressed as:
$$F(T)=1-e^{-(\frac{T}{\xi})^\beta}\, ,$$
where $T$ denotes the time of failure (or the lifetime) of the battery.
\begin{table*}[!t]
  \caption{Failure times of test samples\cite{eom2007life} }
  \centering
  \begin{tabular}{c c c c c c c c c c c c c}
    \hline\hline                   
    Specimen number& 01 & 02 & 03 & 04 & 05 & 06 & 07 & 08 & 
                                                             09 & 10 & 11 & 12 \\[0.5ex]  \hline      
    Failure time & 1373  & 1470  & 1520  & 1427  & 892  & 814  & 777  & 637 & 927  & 688  & 857  & 866 \\ [1ex] 
    \hline
    \hline
    
  \end{tabular}
  \label{Table: Battery_data}
\end{table*}

In Table~\ref{Table:
  Battery_data}, experimental data based on life cycle  tests are obtained from
\cite{eom2007life}. By fitting the data in Table~\ref{Table: Battery_data} to
the parameters of the Weibull distribution, we obtain the corresponding maximum
likelihood estimator (MLE) for
$\xi$ and $\beta$ are $\hat{\xi}=1138$ and $\hat{\beta}=3.55$, respectively.
Now, we consider this MLE Weibull distribution as the baseline model
$P$, which is a data-driven  approximation to the unknown true model.
Next we consider the family of alternative models within a fixed
tolerance $\eta^2$, namely the non-parametric family of models  $\efam$, see \eqref{Qfamily}. This family  accounts for unknown features not necessarily captured in
the baseline model which  was  arbitrarily assumed
to be Weibull.   Furthermore, the family $\efam$ can account for perturbations in the
baseline model---constructed based on the specific dataset in Table~\ref{Table: Battery_data}---due to additional data that may become available or for any errors in the  data used in the MLE step.

Next, we assess the impact of model uncertainty within the family of models $\efam$ on two QoIs associated with lifetime probabilities of the batteries:
\begin{align}
  f_1(t)&=1_{\{0\leq t\leq T\}}(t),t>0,\label{eq:ind-fun}\\
              f_2(t;w)&=\frac{1}{1+e^{w(t-T)}}, t>0 \label{eq:approx-ind-fun}.
\end{align}
The function $f_2(t;w)$ is a commonly used smooth approximation to the indicator
function $f_1(t)$ and is usually referred as  the logistic function, see Section~39.1 of
\cite{MacKay:2003}). The parameter $w$, $w\geq 1$, controls the smoothness of
the approximation.
The QoI  for  the life-time probability is defined exactly as $F_P[T]:=E_P[f_1(t)]=P(0\le t \le
T)$ or through the smooth approximation $\Ep[f_2]$. 


Since the QoI $f_1(t)$ is bounded in  $[0,1 ]$, we can  apply the
Bennett~\eqref{eq:bennet-mgf-ineq}, Bennett-(a,b)~\eqref{eq:bennet-(a,b)} and Hoeffding bounds~\eqref{eq:Hoeffding}) to obtain the uncertainty region, where $a=0$,
$b=1$, $\tilde{a_1}=-F_P(T)$, $\tilde{b_2}= 1-F_P(T)$ and
$\sigma_B^2=Var_P[f_1(t)]$, the latter needed just in the Bennet bound. For $f_2$, we estimate $\Ep[f_2]$ by sampling from
$P$, thus computing  $\mu_2=\Ep[f_2]$, needed in both Bennett bounds. Then, $\ti a_2=-\mu_2$ and $\ti
b_2=1-\mu_2$. In Figure \ref{Fig:Battery} we  compare the lifetime
probabilities given by $f_1$ and $f_2$, where for the latter we set $w=5$. 
In this Figure, we also observe that  the logistic function $f_2$
gives a good approximation of the indicator function $f_1$ since lifetime
probabilities based on them are almost the same. Moreover, we set $\eta^2=0.1$
and also plot the GO divergence bounds of  Theorem~\ref{thm:GO-bounds} based on $f_1$ and Bennett-(a,b) bounds based on $f_2$. We
notice that the bounds almost coincide. We also consider the Bennett-(a,b) bounds
based on a smaller tolerance $\eta^2=0.01$. As we see in the figure, we obtain a significantly 
narrower model bias region.
\begin{figure}
  \centering
  \includegraphics[width=0.5\textwidth]{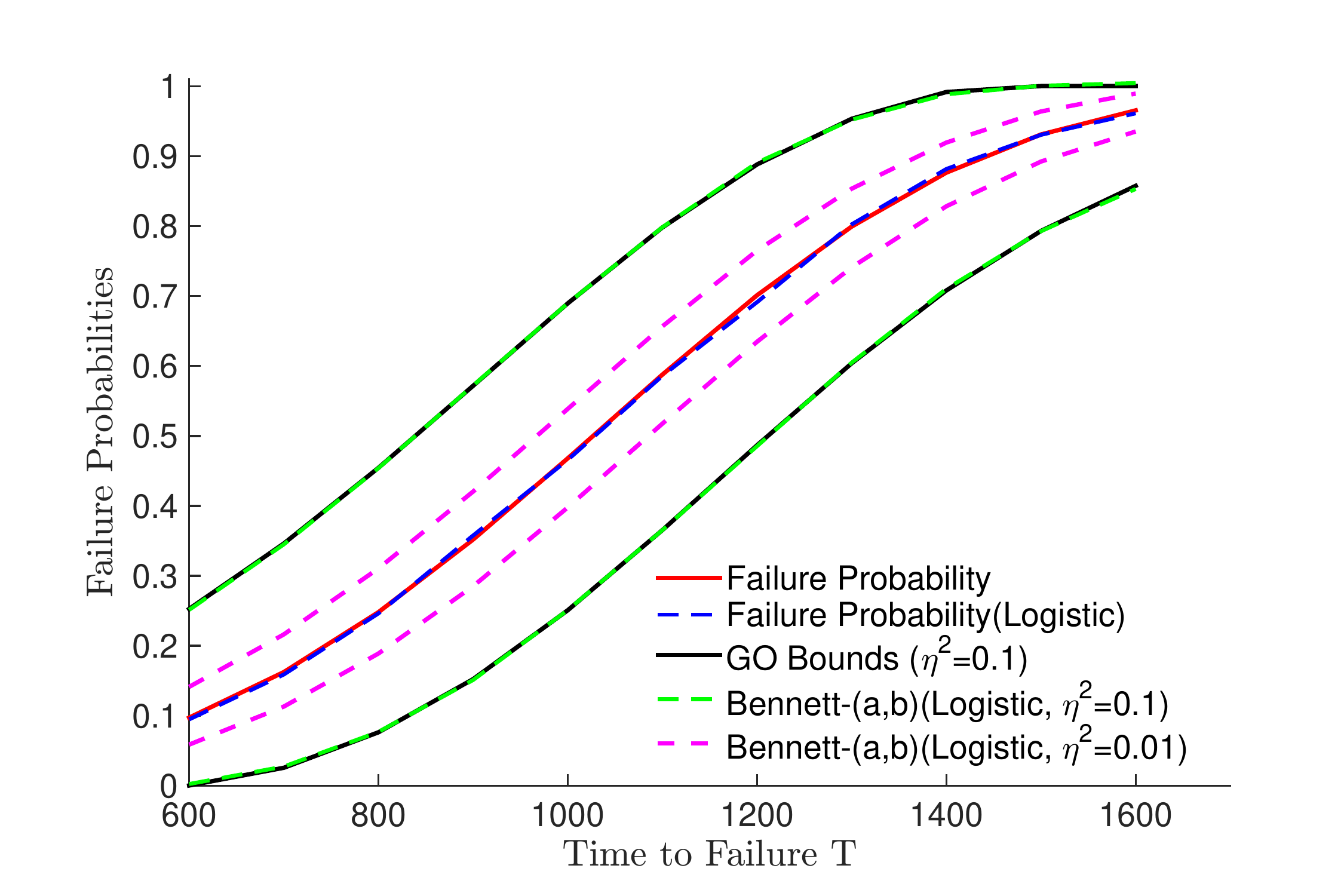}
  \caption{ The blue line is the failure probability based on the logistic function
    $f_2$; The red line is the failure probability based on the indicator function $f_1$; The black
    lines are the GO bounds based on $f_1$ with $\eta^2=0.1$; The green lines
    are the Bennett-(a,b) bounds based on $f_2$ with model uncertainty $\eta^2=0.1$.
    The magenta lines
    are the Bennett-(a,b) bounds based on $f_2$ with $\eta^2=0.01$.}
  \label{Fig:Battery}
\end{figure}

\begin{remark}[Why concentration/information inequalities?]
  \label{battery_remark}
  {\rm 
  As shown in Lemma~2.11, Equation~(2.28) of \cite{katsoulakisuq16}, the $c^*$
  that solves the optimization problem of the GO divergence bound in
  Equation~(\ref{eq:intro:godiv:2}) behaves like
  \begin{align}
    \label{eq:cstar}
    c^*=c_1\eta + O(\eta^2), 
  \end{align}
  for some explicit constant $c_1$ and $\eta^2=\R QP$. Due to~\eqref{eq:cstar} and since estimator variance for the MGF increases
  exponentially with $c$, a larger
  uncertainty threshold $\eta$ will quickly make the accurate estimation of $\Mp{c^*;f}$
  more demanding, as is readily clear from Table~\ref{tab:estimator-cost}
  and \eqref{eq:cstar}. This drawback becomes  especially problematic  when sampling from $P$ is
  computationally expensive, \eg requires MCMC sampling, $P$ is multi-modal,
  \etc, see also the Markov Random Field example in Section~\ref{sec:loc-pert-stat-mod}, where
  sampling challenges can become more pronounced in higher dimensions. Even when $P$ is simple to sample,
  as is the case with the baseline models in \cite{li2012computation} and
  here, 
  avoiding the estimation of
  $\Mp{c^*;f}$ in the GO divergence can still save significant computational time, as Table~\ref{tab:estimator-cost} strongly suggests.  For instance, the Bennett-(a,b)
  bound in~(\ref{eq:bennet-(a,b)}) only requires
  \begin{enumerate*}
  \item[(a)]  the bounds of the QoI, $a,b$, and 
  \item[(b)]  the expected value of the QoI with respect to $P$.
  \end{enumerate*}
  }
\end{remark}

\subsection{Uncertainty Quantification for Markov Random Fields}
\label{sec:loc-pert-stat-mod}
Here we consider the impact on QoIs of localized perturbations to statistical probability distributions of Markov
Random Fields \cite{MacKay:2003} such as Gibbs measures. Such distributions are inherently high-dimensional, allowing us to focus on this aspect of model bias bounds.
In particular,  we consider Gibbs
measures for particle systems defined on  a fixed finite subset $\Lambda_N$ of the infinite dimensional lattice $\mathbb{Z}^d$. Specifically we consider  $\Lambda_N=\{x\in \mathbb{Z}^d,
|x_i|\le n\}$  the square lattice with $N=(2n+1)^d$ lattice sites, where typically $n\gg 1$. 
Before we describe the model, we will specify some necessary notation: we let $S$ be the
configuration space of a single particle at a lattice  site $x \in \mathbb{Z}^d$. For example in a lattice gas model $S=\{0, 1\}$, i.e. the lattice site can be empty or occupied,  and in a Potts model $S=\{0, 1,..., q\}$, i.e. the site is empty or occupied by particles of $q$ different species. In Ising magnetization models studied below, we have that  $S=\{-1, 1\}$, corresponding to down or up spins respectively. 
Then $S^X$ is the configuration space for the particles in any subset  $X \subset
\mathbb{Z}^d$; we denote by $\sigma_X \,=\, \{ \sigma_x\}_{x\in X}$ an element
of $S^X$.
                     %
Next, in order to
define a Gibbs measure on $\Lambda_N$, we first specify the Hamiltonian $H_N(\sigma_{\Lambda_N})$ of a set
of particles in the region $\Lambda_N$. An interaction $\Phi= \{\Phi_X: X \subset \mathbb{Z}^d, X\;\; 
\textrm{finite}\}$  associates to any finite subset $X$ a function
$\Phi_X(\sigma_X)$ which depends only on the particle configuration in $X$ and accounts for all particle interactions within $X$, see
\cite{simon2014statistical} for details. Given an interaction $\Phi$ we then
define the Hamiltonian $H^\Phi_N$ (with free boundary conditions) by
\begin{equation}\label{Hamiltonian:Sec5}
  H_N^\Phi(\sigma_{\Lambda_N})=\sum_{X \subset \Lambda_N}\Phi_X(\sigma_X),
\end{equation}
and Gibbs measure $\mu_N^\Phi$ by
\begin{equation}
\label{eq:Gibbs_exp}
  d\mu_N^\Phi(\sigma_{\Lambda_N}) = \frac{1}{Z_N^\Phi} e^{ - H_N(\sigma_{\Lambda_N})} dP_N(\sigma_{\Lambda_N}) ,
\end{equation}
where $P_N$ is the counting measure on $S^{\Lambda_N}$ and $Z_N^\Phi= \sum_{
  \sigma_{\Lambda_N}} e^{ - H_N(\sigma_{\Lambda_N})}$ is the normalization
constant, also known as the partition function, \cite{simon2014statistical}.

Here we consider classes of perturbed  models with corresponding  interaction $\Psi$ that includes only local perturbations to the interaction
$\Phi$, e.g. \textit{local defects} encoded in the interaction potential $J$, or localized perturbations to the external field $h$ in the example of the 
Ising-type Hamiltonian \eqref{eq:Ising:Hamilt}. We also note that defects of finite temperature 
multi-scale probability distributions
are a continuous source of interest in the
computational materials science community, see, for instance, \cite{Ortiz}; in fact,  
lattice probability distributions such as \eqref{eq:Gibbs_exp}, constitute an important class of
simplified prototype problems. In the case of localized perturbations to the
interaction $\Phi$ in~\eqref{Hamiltonian:Sec5}, the Hamiltonians scale
as follows:
$$
H_N^\Psi(\sigma_{\Lambda_N})=H_N^\Phi(\sigma_{\Lambda_N})+O(1)\, .
$$
Thus the corresponding relative entropy satisfies
\begin{equation}
\begin{aligned}
 R(\mu^\Psi_N\mid\mid \mu^\Phi_N)&=\log E_{\mu_N^\Psi}(e^{\Delta H})+E_{\mu_N^\Psi}( -\Delta H)\\&=O(1)\, ,
  \label{eq:Isng_KL}
\end{aligned}
\end{equation}
uniformly in the system size $N$, where we define $\Delta H=H_N^\Psi - H_N^\Phi$. However,
in most cases, we do not know the exact local perturbation as well as the perturbed Gibbs
measure $\mu_\Psi$. Instead,  based on\eqref{eq:Isng_KL} we can consider a family of perturbed models:
$$\efam=\{\mu_\Psi: R(\mu_\Psi\|\mu_\Phi)\le \eta^2 \}\, .$$
This family   will include  any
perturbation \textit{within that tolerance} $\eta^2$, for example: defects
located at different lattice sites, and of different magnitudes, as the scaling \eqref{eq:Isng_KL} demonstrates rigorously.
\begin{figure}[b]
\centering
\subfloat[]{\includegraphics[width=0.5\textwidth]{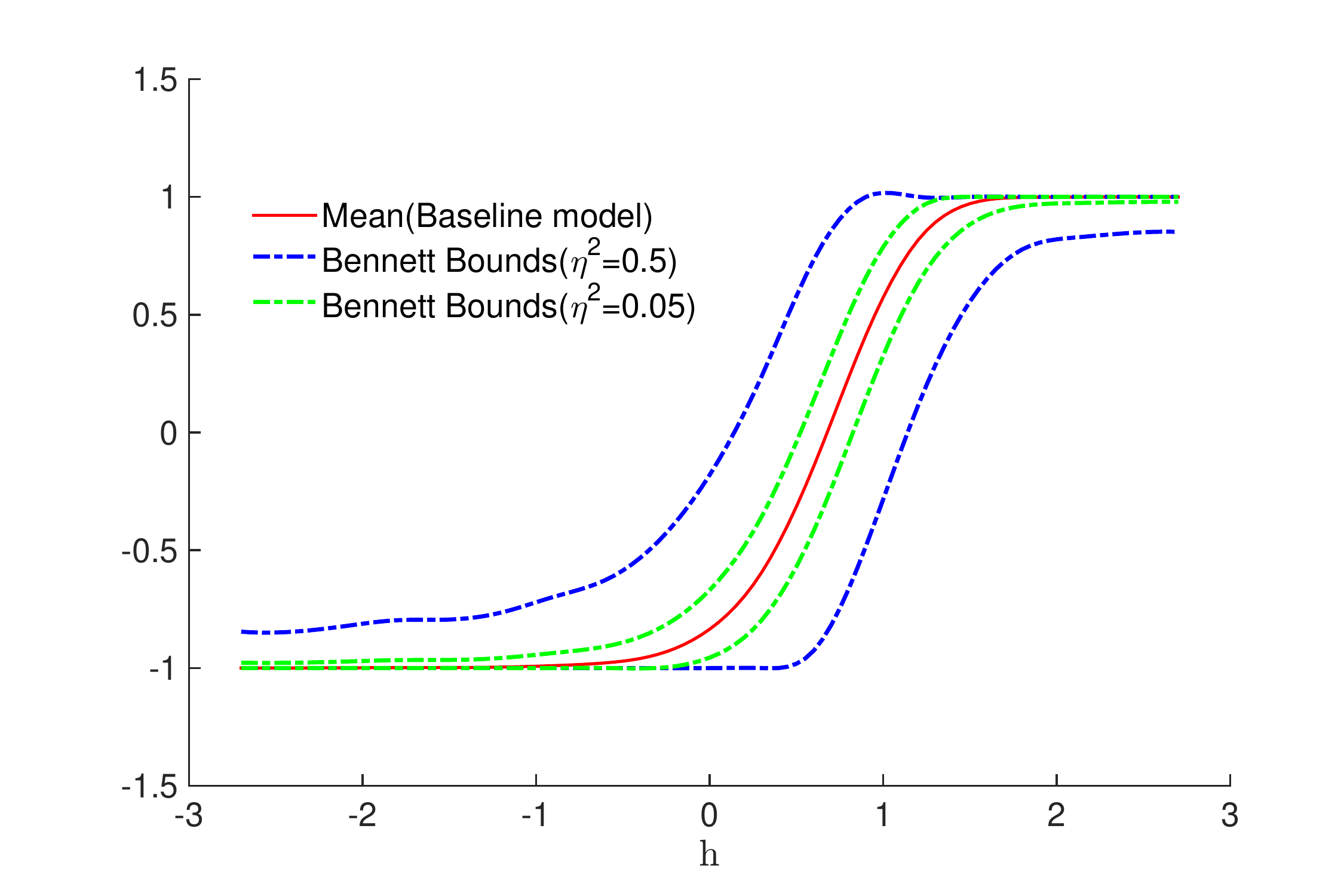}%
\label{Fig:Ising_KL1}}
\hfil
\subfloat[]{\includegraphics[width=0.5\textwidth]{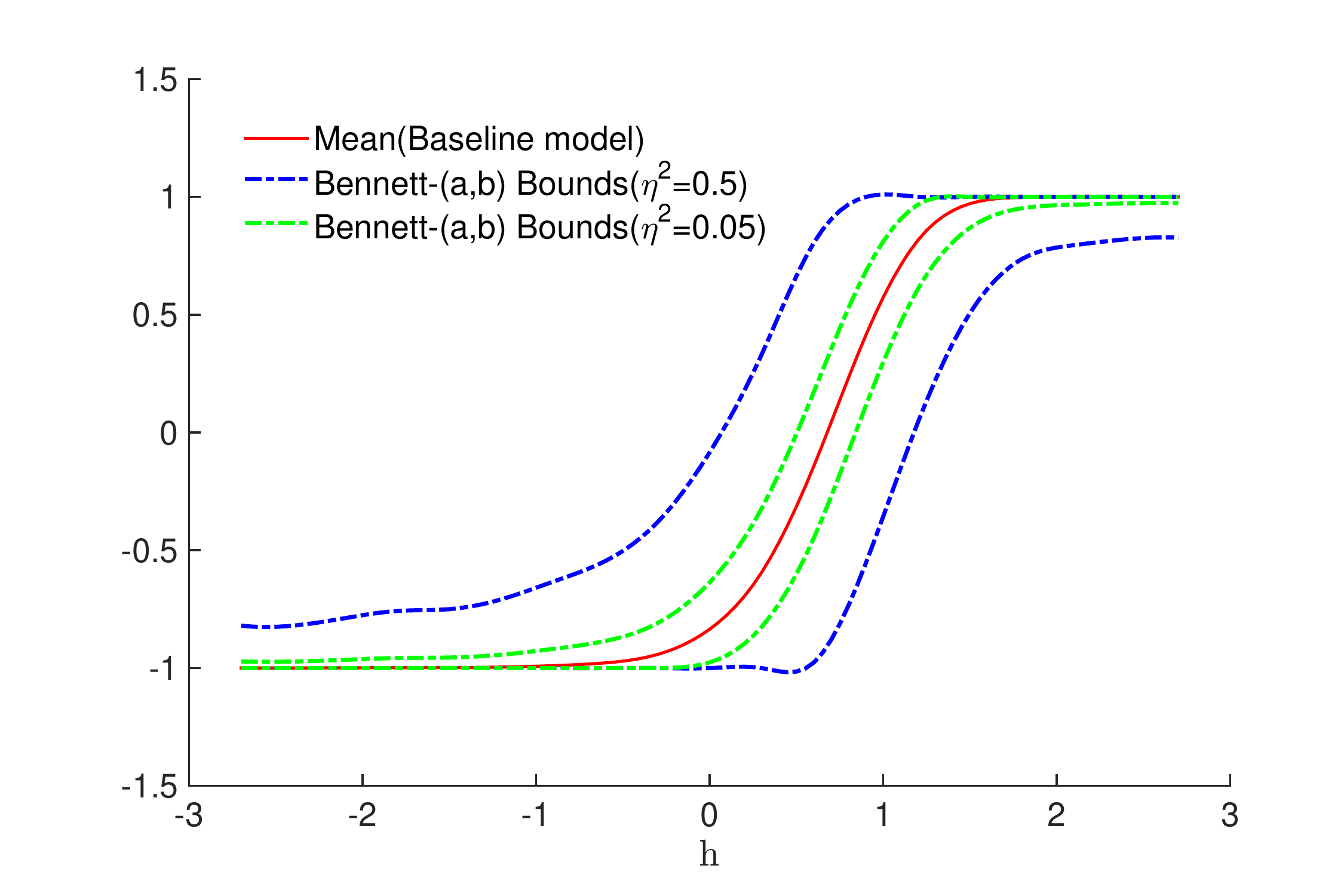}%
\label{Fig:Ising_KL2}}
\caption{\protect\subref{Fig:Ising_KL1}  The red line is the mean of the QoI \eqref{eq:Ising:observable} with $m=1$
for the baseline model \eqref{eq:Ising:Hamilt} with  $J=1$ and
    $\beta=1$.
    The green and blue lines are the Bennett bounds for $\eta^2=0.05$
    and $\eta^2=0.5$, respectively.
\protect\subref{Fig:Ising_KL2} The red line is the mean of the same QoI 
for the baseline model \eqref{eq:Ising:Hamilt} with  $J=1$ and
    $\beta=1$ ; The green and blue lines are the Bennett-(a,b) bounds for $\eta^2=0.05$
    and $\eta^2=0.5$, respectively. In both figures the lattice size is N=100. }
\label{Fig:Ising_KL}
\end{figure}

As a concrete  example of a Hamiltonian\eqref{Hamiltonian:Sec5}, we consider $\mu_\Phi$ to be  a one-dimensional  Ising model probability
distributions on the one-dimensional lattice $\Lambda_N$, labeled
successively by $x=1,2,...,N$. To each site corresponds a spin $\sigma(x)$, with two
possible values: $+1$ or $-1$. The Hamiltonian is given by
\begin{equation}
\label{eq:Ising:Hamilt}
  H^\Phi_{N}(\sigma_{\Lambda_N})=-\beta \sum_{x=1}^{N-1}J(x)\sigma(x)\sigma(x+1)-\beta h\sum_{x=1}^N\sigma(x).
\end{equation}
Using the  concentration/information  inequalities developed in Section~\ref{sec:conc-ineq}, we can obtain model bias  bounds for 
QoIs, such as the localized average around any lattice site $x$,
\begin{equation}
  f(\sigma_{\Lambda_N})=\frac{1}{2m+1}\sum_{\{y: |y-x|\le m\}} \sigma(y)\, ,
  \label{eq:Ising:observable}
\end{equation}
for a fixed radius $m$. In the demonstration below we we pick $m=1$ for concreteness. 
Since the QoI $f$ \eqref{eq:Ising:observable} is bounded, $-1 \le f \le 1$, we can use
the Bennett-(a,b) bound~\eqref{eq:bennet-(a,b)}. Alternatively, we can use the   Bennett bound~\eqref{eq:bennet-mgf-ineq}), which however requires estimating in addition to 
$E_{\mu_N^\Phi}[f]$, the variance  $Var_{\mu_N^\Phi}[f]$ by sampling from $\mu_\Phi$, see also Section~\ref{sec:bounded-obs}. The latter is not unreasonable given that variance computations are necessary in many applications because  they ensure suitable confidence intervals for the averaged QoIs. 
In Figure \ref{Fig:Ising_KL}, we implement
both Bennett and Bennett-(a,b) bounds by considering two different KL divergence tolerances, $\eta^2=0.5$, $0.05$.
A comparisons between   Figure \ref{Fig:Ising_KL1} and Figure \ref{Fig:Ising_KL2} indicates that Bennett  and Bennett-(a,b) bounds are fairly close for this example.

Notable computational advantages  of these concentration/information inequalities over  direct numerical simulation of  alternative models  $Q=\mu_\Psi$, 
as well as over the  GO divergence bounds  in Theorem~\ref{thm:GO-bounds} are the following: (1) when using Theorem~\ref{thm:generalized-go-bounds} along with Bennett-type bounds  \eqref{eq:bennet-mgf-ineq} or \eqref{eq:bennet-(a,b)}, we can deploy computational resources to estimate
$E_{\mu_N^\Phi}[f]$ or possibly $\mbox{Var}_{\mu_N^\Phi}[f]$---see also Table~\ref{tab:MGF-bounds}---just  for the baseline model $P=\mu_\Phi$,
instead of simulating all alternative models $Q=\mu_\Psi$ models; (2) we do not need to use 
the  full GO divergence bounds  in Theorem~\ref{thm:GO-bounds}, which  require potentially expensive  full MGF calculations,  also recalling Remark
\ref{battery_remark}.


\section{Conclusion}
In this paper we combined the uncertainty quantification  information  inequality of \cite{dupuis2011uq,katsoulakisuq16,katsoulakisjcp2017}  together with classical concentration inequalities \cite{concentration} to obtain  
easily implementable bounds for the model bias of quantities of interest (QoIs). The bounds control the model 
bias  in terms of the relative entropy 
between  different models and intrinsic statistical quantities associated to the QoIs in a baseline model, e.g. mean, variance,  $L^\infty$ bound. 
Our results improve substantially on  classical information bounds such as the Pinsker inequality. 
First, our bound scales correctly with the size of the data sets/number of degrees of freedom while 
classical inequalities do not, see Remark~\ref{rem:poor:scalab}. This scaling property is illustrated in Section \ref{sec:scalability} where we discuss  bias bounds for general statistical estimators.
In addition,  we demonstrate the tightness of our bounds in Sections \ref{sec:outline} and \ref{sec:sharpness-and-robustness}: given suitable families of QoIs and a family of models whose Kullback-Leibler divergence with respect to a given baseline   model is less than a tolerance $\eta^2$, there  always exists a QoI and models which saturate the upper and lower bounds. This demonstrates rigorously the precise  sense 
our model bias bound is optimal.  In forthcoming work we will apply and generalize our results to quantify model bias between different stochastic dynamics, e.g.  
Markov processes, in their  long-time regime,  bias   in  phase diagrams of Gibbs-Markov random fields, as well as model bias of coarse-grained models for  equilibrium and non-equilibrium molecular dynamics built via variational inference methods, \cite{Shell,KP2016}.


%

\appendices
\section{Proof of Theorem~\ref{thm:tightness}}
\label{sec:Appendix:tightness:GO}

While all the ingredients in the proof of 
Theorem~\ref{thm:tightness} are already  
present in~\cite{dupuis2011uq,katsoulakisuq16},
(see in particular \cite{katsoulakisuq16}[Theorem 2.9]), 
its formulation is new and we provide here a proof for completeness.  
Theorem~\ref{thm:tightness} follows immediately from following lemma

\begin{lemma}
Let $P$ be a probability measure and 
$f$ to be a non-constant function such 
that its moment generating 
function $\Mp{c;\ti f}$ is finite in a neighborhood of $0$. Let $Q$ be such that $\R QP=M$.  
\begin{enumerate}
\item For any $M \ge0$ the optimization problems
\begin{equation}
\go Q P {\pm f} =\inf_{c>0} \frac{\log M_P(\pm c; \ti f)+M}{c}
\end{equation}
have unique minimizers $c^{\pm}\in[0, + \infty]$. Moreover there exists
$0< M_{\pm}\le\infty$ such that the minimizers $c^{\pm}=c^{\pm}(M)$
are finite for $M \le M_{\pm}$ and $c^{\pm}(M) = + \infty$ if $M >
M_{\pm}$.

\item If $c^\pm=c^{\pm}(M)$ is finite
\begin{equation}
\label{eq:Bpmeqn}
\begin{aligned}
\go Q P {\pm f}  = \frac{\log M_P(\pm c^\pm; \ti f)+M}{c^\pm}
=\frac{d\log M_P}{dc}(\pm c^{\pm}; \ti f)  
=\pm(\E_{P^{\pm c^{\pm}}}[f]- \E_P[f])\,,
%
%
\end{aligned}
\end{equation}
where $c^{\pm}(M)$ is strictly increasing in $M$ and is determined by the
equation
\begin{equation}
\R  {P^{\pm c^{\pm}}}{P}
=M\,.\label{eq:cond:re}%
\end{equation}

\item If $M_{\pm}<\infty$ then $f$ is necessarily $P$ almost surely bounded above/bounded below respectively  with upper/lower bound $f_\pm$.
For $M>M_{\pm}$ we have that $c^{\pm}(M)=+\infty$ and 
\begin{equation}
\go Q P {\pm f}
= \pm (f_\pm - E_P[f])
\,.
\label{eq:fplus}
\end{equation}

\end{enumerate}
\end{lemma}
\begin{proof}
For notational ease, in the proof, let 
us set $H(c)=\log M_P(c; \ti f)$ 
and note that since $\ti f$ is 
centered we have $H(0)=H'(0)=0$. 
We have $H'(c)= \E_{P^c}[f]-\E_p[f]$ and
$H''(c)=\mathrm{Var}_{P^c}(f) >0$ since 
$f$ is not constant $P$ almost surely 

If $d_+ < \infty$ then we have 
$\lim_{c \to d_+} H(c)= \infty$ and 
$\lim_{c \to d_+} H'(c)= \infty$. 
If $d_+ = \infty$ then 
\begin{equation}
\lim_{c \to \infty} H'(c) = \left\{ 
\begin{array}{cl} f_+ - \E_P[f]& \textrm{ 
if }  f \textrm{ is bounded} \\
+\infty & \textrm{ otherwise} 
\end{array}
\right. \,.
\end{equation}
Since $c^{-1}H(c) = c^{-1}\int_0^c 
H'(t) dt$ and $H'(c)$ is strictly 
increasing $c^{-1} H(c)$ is a strictly 
increasing function and we have 
$\lim_{c\to \infty} c^{-1} H(c) = 
\lim_{c\to \infty} H'(c)$ 
which is finite if only if $f$ is bounded.   
Let us set 
\[
B(c; M)=\frac{\log M_P(c; f)+M} {c}=\frac{H(c)+M}{c}
\]
and then distinguish two cases: 

\smallskip 
\noindent
(a) If $d_+ \le \infty$ or if $d_+=\infty$ and $f$ is unbounded then we have $\lim_{c \to 0} B(c; M)= \lim_{c \to d_+} B(c; M) = +\infty$ \
and thus $B(c; M)$ has at least one 
minimum for some $0< c < d_+$. 
By calculus the minimum must be a solution of 
\[
0 = \frac{\partial}{\partial c} B(c;M)= \frac{c H'(c)- H(c) - M }{c^2}
\]
that is me must have $c H'(c)- H(c)= M$. 
Since $\frac{\partial}{\partial c}(c H'(c)- H(c)) = 
cH''(c)> 0$ the function $c H'(c)- H(c)$ is 
strictly increasing and thus there is a 
unique minimizer $c_+$ for $B(c;M)$.    

\smallskip 
\noindent
(b) If $d_+ =\infty$ but $f$ is bounded, since $c H'(c)- H(c)$ is strictly increasing we have $\lim_{c \to \infty} 
cH'(c)- H(c) =M_+$ which may or may not be finite depending on $P$. If $M \le M_+$ we can proceed as in (a) to find a unique minimizer for a finite $c_+$,  while if 
$M> M_+$, $B(c;M)$ is strictly decreasing and thus the minimizer is attained at $c_+=\infty$.   

To conclude the proof we note that if 
$c_+ < \infty$ then  
$c_+ H'(c_+)- H(c_+) =M$
and thus 
\[
B(c_+, M) = H'(c_+)
\]
which proves \eqref{eq:Bpmeqn}. On the other hand a simple computation shows that
for any $c$ 
\[
\R {P^c} P =  c H'(c)- H(c)
\]
and this establishes \eqref{eq:cond:re}. 
Finally if $c_+=\infty$ the infimum is equal to $\lim_{c\to \infty}\frac{H(c)}{c}$ and this establishes \eqref{eq:fplus}.

\end{proof}

\section*{Acknowledgment}
The research of KG and MAK  was partially supported by the Office of Advanced Scientific Computing Research, U.S. Department of Energy, under Contract No. DE-SC0010723.
The research of MAK and LRB was partially supported by the National Science Foundation (NSF) under the grant DMS-1515712. 
The research of JW was partially supported by the Defense Advanced Research Projects Agency (DARPA) EQUiPS program under the grant  W911NF1520122.

\ifCLASSOPTIONcaptionsoff
  \newpage
\fi



\bibliographystyle{IEEEtran}
\bibliography{commonRefs}
\end{document}

%% file: symbols.tex
 \usepackage{amsmath,amssymb,amsthm}
 \usepackage{graphicx}
 \usepackage{xspace}

 \newtheorem{theorem}{Theorem}
\newtheorem{corollary}[theorem]{Corollary}
\newtheorem{definition}[theorem]{Definition}
\newtheorem{example}[theorem]{Example}
\newtheorem{lemma}[theorem]{Lemma}
 \newtheorem{remark}[theorem]{Remark}
 \newtheorem{proposition}[theorem]{Proposition}

\renewcommand{\eqref}[1]{(\ref{#1})}

\newcommand{\sib}{\sigma_{B}}

\newcommand{\E}{\mathbb{E}}
\newcommand{\Ep}{\E_{P}}
\newcommand{\Eq}{\E_{Q}}

\newcommand{\ti}[1]{\tilde{#1}} 

\newcommand{\var}{\mathrm{var}}
\newcommand{\Vp}{\var_{P}}
\newcommand{\Vq}{\var_{Q}}

\newcommand{\eg}{e.g.,\xspace}

\newcommand{\etc}{etc.\@\xspace}
\newcommand{\ie}{i.e.,\xspace}

\newcommand{\R}[2]{R(#1\|#2)} 

\newcommand{\go}[3]{\Xi(#1\|#2;#3)} 

\newcommand{\Up}[2]{U_{+}(#1;#2)}
\newcommand{\Um}[2]{U_{-}(#1;#2)}
\newcommand{\Upm}[2]{U_{\pm}(#1;#2)}

\newcommand{\bias}[1][f]{\Eq[#1]-\Ep[#1]} 

\newcommand{\Phip}[1]{\Phi(#1)} 
\newcommand{\Mp}[1]{M_{P}(#1)} 
\newcommand{\Ptil}[1]{P^{c_{#1}}} 

\newcommand{\Cp}[1]{\log \Mp{#1}} 

\newcommand{\efam}{\mathcal{Q}_{\eta}}

\newcommand{\Fp}{\mathcal{F}_{P}} 

\newcommand{\nn}[1]{#1_{1},\ldots,#1_{n}}

\newcommand{\sepeq}[1]{
\medskip
\noindent 
{\bf #1} 
}
